\documentclass[sigconf, nonacm]{acmart}
\newif\ifarxiv
\arxivtrue 
\newcommand{\shortlong}[2]{\ifarxiv{#2\xspace}\else{#1\xspace}\fi}
\newif\ifrebuttal
\rebuttaltrue
\usepackage[normalem]{ulem}
\usepackage{cancel}

\newif\ifdiff
\difffalse

\newcommand{\rebuttal}[2]{%
  \ifrebuttal
    \ifdiff
      {\color{red}\sout{#1}} {\color{blue}#2}\xspace%
    \else
      {\color{black}#2}\xspace%
    \fi
  \else
    #1\xspace%
  \fi
}

\newcommand{\rebuttalsubsubsection}[2]{%
  \ifrebuttal%
    \ifdiff%
      \subsubsection{\textcolor{red}{\sout{#1}}\textcolor{blue}{#2}}%
    \else%
      \if\relax\detokenize{#2}\relax%
      \else%
        \subsubsection{\textcolor{blue}{#2}}%
      \fi%
    \fi%
  \else%
    \subsubsection{#1}%
  \fi%
}

\usepackage{color}
\usepackage{xargs} 
\usepackage{xspace} 

\usepackage{algorithm}
\usepackage{algorithmic}
\usepackage{listings}
\usepackage{multirow}
\usepackage{amsmath}
\usepackage{breqn}
\usepackage{array}
\usepackage{subcaption}

\usepackage{apxproof}
\newtheorem{theorem}{Theorem}

\newtheoremrep{inequality}{Inequality}
\newtheoremrep{theorem}{Theorem}
\newtheoremrep{lemma}{Lemma}
\newtheoremrep{definition}{Definition}
\newtheoremrep{proposition}{Proposition}

\usepackage{enumitem}
\allowdisplaybreaks
\def\notationcolor{black} 
\newcommand{\notation}[2]{\newcommand{#1}{{\textcolor{\notationcolor}{\ensuremath{#2}}}}}

\newcommand{\ourmech}{Ours-Exp\xspace}

\notation{\mech}{\mathcal{M}} 
\notation{\mechset}{\mathcal{W}} 
\notation{\outp}{\omega} 
\notation{\lv}{v} 
\notation{\LL}{LL} 
\notation{\lvset}{\mathcal{V}} 
\notation{\range}{\text{range}} 
\notation{\epsDP}{\epsilon_{\mathrm{DP}}} 
\notation{\epsPuffer}{\epsilon_{\mathrm{puffer}}} 
\notation{\highinf}{\mathcal{H}} 
\notation{\lowinf}{\mathcal{L}} 
\notation{\randalg}{\mathcal{A}} 

\notation{\secret}{s} 
\notation{\secretset}{S} 

\notation{\prior}{\theta} 
\notation{\priorset}{\Theta} 
\notation{\ab}{a(b)}

\notation{\data}{\mathcal{D}} 
\notation{\numdata}{n} 
\notation{\universe}{\mathcal{U}} 
\notation{\Data}{\textit{Data}} 
\notation{\xor}{\oplus}
\notation{\neighborset}{\mathcal{N}} 
\notation{\dataquilt}{\mathcal{D_L}} 
\notation{\dataneighbor}{\mathcal{D_N}} 
\notation{\Dataquilt}{\textit{Data}_{\mathcal{L}}} 

\notation{\real}{\mathcal{R}} 
\usepackage[createShortEnv,conf={restate,no link to proof}]{proof-at-the-end}

\begin{document}
\title{Composition for Pufferfish Privacy}

\author{Jiamu Bai, Guanlin He, Xin Gu, Daniel Kifer, Kiwan Maeng}
\affiliation{%
  \institution{Penn State University}
}

\begin{abstract}
When creating public data products out of confidential datasets, inferential/posterior-based privacy definitions, such as Pufferfish, provide compelling privacy semantics for data with correlations. However, such privacy definitions are rarely used in practice because they do not always compose. For example, it is possible to design algorithms for these privacy definitions that have no leakage when run once but reveal the entire dataset when run more than once. We prove necessary and sufficient conditions that must be added to ensure linear composition for Pufferfish mechanisms, hence avoiding such privacy collapse. These extra conditions turn out to be differential privacy-style inequalities, indicating that achieving both the interpretable semantics of Pufferfish for correlated data and composition benefits requires adopting  differentially private mechanisms to Pufferfish. We show that such translation is possible through a concept called the $\ab$-influence curve, and many existing differentially private algorithms can be translated with our framework into a composable Pufferfish algorithm. We illustrate the benefit of our new framework by designing composable Pufferfish algorithms for Markov chains that significantly outperform prior work. 
\end{abstract}

\maketitle
\renewcommand{\shortauthors}{} 
\renewcommand{\shorttitle}{}   

\shortlong

\section{Introduction}
Differential privacy (DP)~\cite{dwork2006calibrating} is a gold standard in privacy protection
and is deployed in many important real-world applications by organizations such as the
U.S. Census Bureau~\cite{tdahdsr,ashwin08:map}, Uber~\cite{FLEX,chorus}, Apple~\cite{appledpscale}, Meta~\cite{fburlshares,yousefpour2021opacus}, Microsoft~\cite{DingKY17}, and Google~\cite{rappor,gboard,tensorflowprivacy}.
This success is not only due to differential privacy's rigorous guarantees against data reconstruction attacks \cite{dwork2006calibrating,dinur2003revealing} but also due to its composition properties \cite{PINQ,kairouz2015composition}. Composition
measures the total privacy loss associated with a collection of mechanisms $\mech_1,\dots, \mech_k$ in terms of their individual
privacy losses. For example, if each $\mech_i$ satisfies differential privacy with parameter $\epsilon_i$, the combined release of all of their outputs satisfies differential privacy with parameter $\epsilon=\sum_{i=1}^k \epsilon_i$ (i.e., the privacy parameter composes linearly). This property is critical because it allows a larger mechanism to be created out of smaller, simpler mechanisms. Without composition, performing a privacy analysis of a complex disclosure avoidance system becomes intractable \cite{jarmin2023depth}.

The semantics of differential privacy use the hypothetical/counterfactual world framework \cite{causal}: for every hypothetical database $\data$ and secret about a person (e.g., the contents of one person's record), there is an associated counterfactual database $\data^\prime$ in which the secret has been completely scrubbed from the database.
This counterfactual database is considered a risk-free baseline for that person, i.e., any data release from the counterfactual database $\data^\prime$ poses no risk to that person. The differential privacy constraints on a mechanism $\mech$ guarantee that for any pair $(\data,\data^\prime)$ of the hypothetical database $\data$ and any of its counterfactual databases $\data^\prime$, inference about a
person based on the output of $\mech(\data)$  would be nearly the same as the  inference based on $\mech(\data^\prime)$.

However, for datasets with correlations, such a guarantee may not be meaningful \cite{nofreelunch,pufferfish}. For example, consider a medical dataset $\data$ in which one person's medical status (e.g., whether the person has an infectious disease) can have a causal effect on others (e.g., they may infect close relatives, who may then go on to infect others). Removing one person's record creates a counterfactual database $\data^\prime$ that is not necessarily risk-free for that person, as other records in the counterfactual database provide strong clues about the missing record. 
For differential privacy to provide a meaningful privacy guarantee in such cases, the risk-free counterfactual database must be the one with all the records removed (one contagious individual can potentially affect all records). 
However, ensuring that a data release from a hypothetical database is similar to a release from this risk-free (empty) database would entirely ruin utility.

\emph{Inferential/posterior-based} privacy definitions have emerged as an alternative privacy concept for such settings where data are correlated~\citep{pufferfish, IP, BDP, PKDP}. 
With these privacy definitions, 
the data curator specifies a set of distributions for the dataset, and it is ensured that a Bayesian adversary using those distributions would learn little about any target individual.
However, posterior-based definitions are difficult to deploy in practice because \emph{they do not guarantee composition}. We show that it is possible to design a mechanism that has zero privacy leakage under posterior-based definitions when run once but has a good chance of revealing the entire database (i.e., no privacy) when run more than once (Section~\ref{sec:discomposition}).

To counter the lack of composition guarantees,
we study the following question, focusing on one of the popular posterior-based privacy definitions, \emph{Pufferfish} \cite{pufferfish}:
what additional constraints need to be added to Pufferfish in order to guarantee linear composition? 
We prove necessary and sufficient conditions that establish a formal link between Pufferfish and differential privacy. That is, in order to 
ensure that the Pufferfish privacy parameter $\epsPuffer$ composes linearly, 
one should design mechanisms that satisfy differential privacy with some $\epsDP$ and translate it to Pufferfish privacy.
We show that such translation is possible through a concept called the $\ab$-influence curve, and existing differential privacy mechanisms can be translated into composable Pufferfish mechanisms with our framework.
Our result generalizes the previous framework of Song et al. \cite{song2017pufferfish}, which showed the Laplace mechanism for differential privacy can be used to achieve composable Pufferfish for Markov networks. 
Thus, our work addresses a major concern hindering the development of query-answering platforms based on Pufferfish---currently, designing a new mechanism requires tedious manual analyses of how it interacts with \textit{all} other supported mechanisms. Linear composition removes this bottleneck, and our theoretical results imply that systems designers should build upon existing DP platforms (and leverage their large library of existing mechanisms) simply by rescaling the privacy parameter.
Our contributions are:
\begin{itemize}[leftmargin=*]
    \item We show that posterior-based definitions generally do not defend against privacy leakage via composition, by constructing mechanisms that have zero leakage under those definitions when run once yet provide no meaningful privacy when run multiple (as few as two) times (Section~\ref{sec:discomposition}).
    \item We prove \emph{necessary conditions} for a Pufferfish mechanism to ensure linear composition.
    Our result shows that a Pufferfish mechanism must meet a class of inequalities that we call the NfC constraints (Section~\ref{sec:conditions}), which subsume the differential privacy inequalities.
    We also show that only a subset of the NfC constraints are post-processing invariant, and those subsets still contain differential privacy-style inequalities but rule out many others. 
    The implication is that Pufferfish should be augmented with differential privacy-style inequalities in order for its privacy parameter $\epsPuffer$ to compose.
    \item We next study \emph{sufficient conditions} for composition---how to construct composable Pufferfish privacy mechanisms---for general tabular datasets (Section~\ref{sec:puffer-mech}). 
    We show that if a mechanism satisfies per-entry $\epsDP$-differential privacy,
    it also satisfies $\epsPuffer$-Pufferfish privacy with a certain $\epsPuffer$, and $\epsPuffer$ will compose linearly. 
    We introduce a concept called an \emph{$\ab$-influence curve}, which links the $\epsPuffer$ parameter to the appropriate $\epsDP$ parameter. This part is a generalization of the work of Song et al. \cite{song2017pufferfish}, whose technique was limited to the Laplace mechanism and Markov network priors. We also study how to make the \ab-curve robust to mis-specification of the data-generating priors.
    \item We empirically compare our approach to prior work on composable Pufferfish mechanisms~\cite{song2017pufferfish} to show that we can answer popularity queries over Markov chains with significantly better accuracy, as our framework allows adapting other differential privacy mechanisms that are better-suited to queries of interest.
\end{itemize}

\label{sec:intro}

\section{Background and Notation}\label{sec:background}
We use $\universe$ as the set of all possible datasets, and $\data\in\universe$ to represent a specific dataset. The dataset will be a table where rows are records and columns are attributes, except for Sections \ref{sec:discomposition} and \ref{sec:conditions}, whose results do not need $\data$ to have any particular structure. 
We use $\prior$ to represent a prior distribution over datasets (e.g., a distribution that an attacker may use to make inferences), and $\priorset$ for a set of prior distributions. We let $\Data$ denote a random variable over datasets, so that $\Pr[\Data=\data~|~\prior]$ is the probability of the data curator having the dataset $\data$ according to the prior distribution $\prior$.

Protecting privacy in the Pufferfish framework \cite{pufferfish} requires specifying  what needs to be protected. A potential secret $\secret$ is a statement, such as ``Bob has cancer'', which should be protected from an attacker. That is, an attacker should have difficulty inferring whether $\secret$ is true or not. We use the notation \emph{$\secret(\data)=\text{true}$} to indicate that in dataset $\data$, secret $\secret$ is true. When potential secrets are specified in (mutually exclusive) pairs, like  $(\secret_1,\secret_2)$, it means the goal is to prevent the attacker from figuring out whether $\secret_1$ is true or whether $\secret_2$ is true \cite{pufferfish}. Similarly, specifying $(\secret,\neg\secret)$ as a pair means the goal is to prevent an attacker from learning whether or not $\secret$ is true.
Let $\secretset$ denote a set of potential secret pairs that a data curator wishes to protect.

A mechanism $\mech$ is a deterministic or randomized algorithm whose input is a confidential dataset $\data$ and whose output is supposed to be a privacy-preserving data product (e.g., data summaries, trained machine learning models, etc.).
We let $\outp$ denote the output of a mechanism. The notation we use is summarized in Table \ref{tab:notation}.

\begin{table}[t]
\begin{center}
\begin{tabular}{|cp{0.75\linewidth}|}\hline
$\data$ & Dataset \\
$\Data$ & A random variable representing the true dataset \\
$\universe$ & Database domain (set of all possible datasets) \\
$\secretset$ & The set of all potential secret pairs  \\
$\secret$ & A potential secret \\
$\prior$ & A prior probability distribution over datasets  \\
$\priorset$ & The set of all possible prior distributions  \\
$\mech$ & A privacy-preserving mechanism or algorithm \\
$\epsPuffer$ & Pufferfish Privacy budget \\
$\epsDP$ & Differential Privacy budget \\
$\outp$ & An output from a mechanism $\mech$  \\
$\dataquilt$ & A low-influence part of $\data$ \\
$\Dataquilt$ & A random variable over $\dataquilt$ \\
\hline
\end{tabular}
\caption{Table of Notation}\label{tab:notation}
\label{tab:notation}
\vspace{-20pt}
\end{center}
\end{table}

\subsection{Differential Privacy}\label{sec:defdp}
A privacy definition is a set of restrictions that a mechanism $\mech$ has to follow. One of the most foundational is $\epsDP$-differential privacy, also known as pure differential privacy.
\begin{definition}[Pure Differential Privacy \cite{dwork2006calibrating}]\label{def:dp}
    Given a privacy parameter $\epsDP\geq 0$, a mechanism $\mech$ satisfies $\epsDP$-differential privacy if, for all pairs of datasets $\data,\data'$ that differ on the contents of one person's record, and all $\outp$ in $\range(\mech)$, 
    \[
\Pr[\mech(\data) = \outp]
\;\leq\;
e^{\epsDP}
\Pr[\mech(\data') = \outp].
\]
(Note: the probability is with respect to the randomness in $\mech$ only.)
\end{definition}
One can think of $\data$ as a hypothetical dataset and $\data^\prime$ as a counterfactual dataset in which the record of one person has been replaced. 
Differential privacy guarantees that inference about a person based on the output $\mech(\data)$ would be similar to inference based on $\mech(\data^\prime)$. Hence, if the person's influence is truly eliminated in counterfactual worlds, her privacy is protected.

Differential privacy has two important properties: \emph{post-processing invariance} and \emph{linear composition}. Post-processing invariance means that if $\mech$ satisfies $\epsDP$-differential privacy, and $f$ is any postprocessing function, the combined function $f\circ\mech$
satisfies differential privacy with the same privacy parameter $\epsDP$.
Linear composition means that if $\mech_1,\dots, \mech_k$ are a sequence of mechanisms, and each $\mech_i$ satisfies $\epsDP_i$-differential privacy, the mechanism $\mech^*$, which runs all of those mechanisms and returns their outputs, satisfies $\sum_{i=1}^k\epsDP_i$-differential privacy.
In this paper, when we say a mechanism ensures composition, we mean it ensures linear composition or better (sub-linear).
A special variant of differential privacy that is particularly useful for this paper is per-entry differential privacy:
\begin{definition}[Per-Entry DP]\label{def:pedp}
    Given a privacy parameter $\epsDP\geq 0$, a mechanism $\mech$ satisfies per-entry $\epsDP$-differential privacy if, for all pairs of tabular datasets $\data,\data'$ that differ on the contents of one attribute of one person's record, and all $\outp$ in $\range(\mech)$, 
    \[
\Pr[\mech(\data) = \outp]
\;\leq\;
e^{\epsDP}
\Pr[\mech(\data') = \outp].
\]
\end{definition}
The difference between Definitions \ref{def:dp} and \ref{def:pedp} is that in the latter, the counterfactual database is formed by only replacing part of a record (the value of one attribute) instead of the entire record.
Definition~\ref{def:pedp} can naturally extend to a group of $k$ entries, in which case the privacy linearly degrades with $k$:
\begin{definition}[Group DP]\label{def:gedp}
    A per-entry $\epsDP$-differential privacy mechanism $\mech$ achieves $k\epsDP$-differential privacy for a group of $k$ entries. That is, if $\data,\data'$ differ by $k$ entries, for all $\outp$ in $\range(\mech)$, 
    \[
\Pr[\mech(\data) = \outp]
\;\leq\;
e^{k\epsDP}
\Pr[\mech(\data') = \outp].
\]
\end{definition}
Group DP can handle correlations in the database by considering all the potentially correlated entries as a protected group,
but doing so often leaves too little utility; for example, when all the entries are possibly correlated (even loosely)---as in our previous example of a person with an infectious disease---group DP degrades the privacy linearly with the number of total entries in the database.

\subsection{Posterior-based Privacy Definitions}
Unlike differential privacy, which treats datasets as deterministic, posterior-based privacy definitions (e.g., \citep{pufferfish, IP, BDP, PKDP}) model attackers as Bayesian reasoners who view the dataset as a random variable. Some of our results concern posterior-based privacy definitions in general, while others concern a specific definition called Pufferfish privacy~\cite{pufferfish}. We first distill common properties of posterior-based privacy, and then describe a special case called Pufferfish privacy.

\begin{definition}[Posterior-based Privacy Properties]\label{def:posterior}
    A posterior-based privacy definition has a privacy parameter $\epsilon\geq 0$, a set $\priorset$ of prior distributions an attacker may use, and a collection of potential secrets $\{\secret_1,\secret_2,\dots\}$.
    The definition must use these items to define a set of restrictions that a mechanism $\mech$ must satisfy.
    Furthermore, if $\Pr(\secret_i~|~\prior)=\Pr(\secret_i~|~\mech(\data),\prior)$ for all $\secret_i$ and all $\prior\in\priorset$, then $\mech$ satisfies the definition with privacy parameter $\epsilon=0$. That is, if the posterior distribution of secrets is the same as the prior distribution, then the assessed privacy leakage is $0$.
    Note that the probability is taken with respect to randomness in both $\mech$ and $\prior$, and larger values of $\epsilon$ represent weaker privacy protections than smaller values of $\epsilon$.
\end{definition}

\begin{definition}[Pufferfish Privacy~\cite{pufferfish}]
\label{def:pufferfish}
Given a set $\priorset$ of priors, a set $\secretset$ of  potential secret pairs, and a privacy parameter $\epsPuffer\geq 0$,
a randomized mechanism $\mech$ satisfies $\epsPuffer$-Pufferfish privacy if, for all pairs of potential secrets $(\secret_i, \secret_j) \in \secretset$ and all priors $\prior \in \priorset$ such that $\Pr[\secret_i | \prior] > 0$ and $\Pr[\secret_j | \prior] > 0$, and for all possible outputs $\outp$,
\[
\Pr[\mech(\Data) = \outp \mid \secret_i, \prior]
\;\leq\;
e^{\epsPuffer}
\Pr[\mech(\Data) = \outp \mid \secret_j, \prior].
\]
\end{definition}
Let $\frac{\Pr(\secret_i~|~\prior)}{\Pr(\secret_j~|~\prior)}$ be the prior odds, specifying how likely the attacker thinks $\secret_i$ is true compared to $\secret_j$ being true before seeing the data. After seeing the data, the attacker forms the posterior odds, $\frac{\Pr(\secret_i~|~\mech(\data)=\outp, \prior)}{\Pr(\secret_j~|~\mech(\data)=\outp,\prior)}$. Pufferfish privacy guarantees that the posterior odds are at most $e^\epsPuffer$
times the prior odds, i.e., the attacker's belief changed by a factor of at most $e^\epsPuffer$. 
If the ratio between the prior and the posterior odds is the same, it means there is no privacy leakage ($\epsPuffer=0$).
Pufferfish satisfies post-processing invariance \cite{pufferfish}, but does not guarantee composition \cite{pufferfish}.
We show that this problem is very severe and can lead to privacy collapse (Section \ref{sec:discomposition}).

\subsection{Trust and Threat Model}  
We assume a trusted data curator who releases query results by running multiple privacy mechanisms on confidential data, and an adversary who attempts to infer protected secrets from the released outputs.
We assume that the adversary can leverage prior information about the data distribution and the protected secret pairs.
The role of the privacy mechanisms is to prevent the released outputs from significantly increasing the adversary's ability to distinguish between the secrets that the curator wishes to protect. 
This is the setting of Pufferfish \cite{pufferfish}, but we additionally consider the combined privacy leakage of multiple mechanisms (composition).

\section{Related Work}\label{sec:related}
\textbf{\textit{Inevitability of DP for Composition.}}
Blanc et al.~\cite{blanc2025differential} recently studied composition in a general axiomatic setting. If a privacy definition satisfies their axioms (some of which are related to composition), then any mechanism for that definition can be replaced by an  $(\epsilon,\delta)$-DP mechanism whose accuracy, on statistical tasks involving \emph{independent} records, is nearly the same. 
However, their work cannot be applied to \emph{correlated data}, and their axioms are restrictive:  pure $\epsilon$-differential privacy violates their composition axioms, and posterior-based privacy definitions almost always violate their symmetry axioms.
Thus, the work targets different privacy definitions and data assumptions from ours.

Generally, composition results outside of differential privacy are limited in scope. Bhaskar et al. \cite{noiseless} studied composition for mechanisms where the only randomness comes from the prior. Their composition result required refreshing large parts of the data (i.e., resampling  from the prior) in between queries. 
An approach by  Farokhi \cite{farokhi2021noiseless} studies a prior-free noise-free privacy definition and proves composition without restrictions, but the definition cannot always prevent a mechanism from leaking secret information.

Although arbitrary Pufferfish mechanisms will not compose \cite{pufferfish}, Song et al. \cite{song2017composition,song2017pufferfish} showed that some specific mechanisms compose with themselves. They considered the set $\priorset$ of priors to be all Markov networks with the same structure and showed that a Laplace mechanism satisfying differential privacy with a privacy parameter $\epsDP$ can satisfy Pufferfish privacy with some other parameter $\epsPuffer$ and will compose. 
Our result is a generalization of this prior result: we propose a framework which allows translating \emph{any} set of (per-entry) $\epsDP$-differential privacy mechanisms into composable $\epsPuffer$-Pufferfish privacy mechanisms for \emph{any} data prior through a concept called an $\ab$-influence curve.

\textbf{\textit{Inferential/Posterior-based Definitions.}}
This work focuses on Pufferfish privacy~\cite{pufferfish}, a popular posterior-based privacy definition that is a generalization of many other posterior-based definitions~\cite{zhang2022attribute, DependDP,BDP,IP}.
Other notable approaches related to posterior-based privacy include: adapting DP to attackers with deterministic knowledge about the data \cite{blowfish}, changing how the dissimilarities between posteriors are quantified \cite{CIP,pierquin2024renyi}, and adapting the counterfactual framework of DP to compare attacker inference about a hypothetical database to a scrubbed database \cite{DistDP,PKDP}.%

\textbf{\textit{Pufferfish Privacy and Composition.}}
Lack of composition in Pufferfish has been sidestepped in several ways. Pierquin et al. \cite{pierquin2024renyi}, Zhang et al. \cite{zhang2025sliced}, and Liang et al. \cite{liang2020pufferfish} performed a joint privacy analysis across multiple runs of mechanisms. Tao et al. \cite{zhang2025differential} proposed a framework for performing such a joint analysis. In many cases \cite{pierquin2024renyi,zhang2025sliced,zhang2025differential}, the measure of dissimilarity in posterior distributions was changed from Definition \ref{def:pufferfish} to mirror Renyi differential privacy \cite{renyidp} or $(\epsilon,\delta)$-differential privacy \cite{ourdata}. An alternative to a joint analysis is to
consider restricted types of mechanisms that compose with themselves. For example, Song et al. \cite{song2017composition,song2017pufferfish} proposed the Markov Quilt Mechanism (MQM) that adapts the Laplace mechanism to Pufferfish with Markov network priors. This mechanism was also used by Shafieinejad et al. \cite{shafieinejad2021privacy} and Cao et al. \cite{cao2017quantifying}, who adjusted the privacy budget of differentially private algorithms to account for temporal correlations in Markov chains.
In other work, the compositional properties were not studied \cite{song2017pufferfish,ding2022kantorovich,ding2024approximation}.

\section{Privacy Collapse}
\label{sec:discomposition}
We next use simple examples to illustrate that posterior-based privacy can catastrophically fail upon composition.
While pathological, the main message is that posterior-based privacy definitions do not have built-in safeguards against privacy collapse arising from composition. Necessary conditions for such safeguards are proved in Section \ref{sec:conditions} and sufficient conditions are presented in Section \ref{sec:puffer-mech}.

\textbf{\emph{Single-secret/prior Example.}} 
We start with one secret and one prior and study to what extent remedies within the posterior-based framework (e.g., adding more potential secrets or adding more priors) are effective. The proofs can be found in the \shortlong{full version \cite{nfcarxiv}}{appendix}.

\begin{example}\label{ex:one}
Consider mutually exclusive potential secrets $\secret_1$ and $\secret_2=\neg\secret_1$. Let $\priorset=\{\prior^*\}$ consists of a single prior. Let $\universe_1$ (resp., $\universe_2$) be the set of datasets for which $\secret_1$ (resp., $\secret_2$) is true and have nonzero probability under $\prior^*$. Suppose $\universe_1$ and $\universe_2$ each contain at least 2 datasets.
Define mechanism $\mech^*$ as follows. On input $\data$, 
\begin{itemize}[leftmargin=*]
    \item if $\secret_1(\data)$ is true, sample $\data^\prime$ from $\Pr(\cdot~|~\secret_2,\prior^*)$ and outputs $(\data, \data^\prime)$,
    \item if $\secret_2(\data)$ is true, sample $\data^\prime$ from $\Pr(\cdot~|~\secret_1,\prior^*)$ and outputs $(\data^\prime, \data)$. 
\end{itemize}
\end{example}
The next theorem shows this mechanism has zero leakage under posterior-based privacy definitions (see Definition \ref{def:posterior}) but can reveal the true database with finitely many (and as few as two) runs.
The intuition is that $\mech^*(\data)$ returns a tuple where the left part is always the same if $\secret_1$ is true and the right part is always the same if $\secret_2$ is true. If $\mech^*$ is run twice, and the output of the first run is different from the output of the second run, the tuple component which didn't change corresponds to the true input dataset.
\begin{theoremE}\label{the:single-prior-puff}
Let $\mech^*$ be the mechanism from Example \ref{ex:one}. Given an output $\outp$ of $\mech^*$,  $\Pr(\secret_1~|~\mech^*(\Data)=\outp,\prior^*)=\Pr(\secret_1~|~\prior^*)$ and $\Pr(\secret_2~|~\mech^*(\Data)=\outp,\prior^*)=\Pr(\secret_2~|~\prior^*)$. Thus, $\mech^*$ has zero leakage under posterior-based privacy definitions. 
However, if the mechanism is run multiple times, the expected number of runs before the input dataset is revealed is $\leq 1+\frac{1}{1-\max_\data \max_{\secret_i}\Pr(\data~|~\secret_i)}$. 
\end{theoremE}
\begin{proofE}
    Note that $\outp$ is an ordered pair of datasets $(\data_1, \data_2)$ where $\secret_1$ is true of $\data_1$ and $\secret_2$ is true of $\data_2$. Furthermore, either $\data_1$ or $\data_2$ is the true input dataset. Next, we note that:
    \begin{align*}
        & \Pr(\mech^* \text{ outputs }(\data_1, \data_2)~|~\secret_1,\prior^*) \\
        & \phantom{==}=\sum_{\data} \Pr(\data ~|~\secret_1,\prior^*)\Pr(\mech(\data)=(\data_1, \data_2)) \\
        & \phantom{==}=\Pr(\data_1 ~|~\secret_1,\prior^*)\Pr(\data_2~|~\secret_2,\prior^*)\\
        & \Pr(\mech^* \text{ outputs }(\data_1, \data_2)~|~\secret_2,\prior^*) \\
        & \phantom{==}=\sum_{\data} \Pr(\data ~|~\secret_2,\prior^*)\Pr(\mech(\data)=(\data_1, \data_2)) \\
        & \phantom{==}=\Pr(\data_1 ~|~\secret_1,\prior^*)\Pr(\data_2~|~\secret_2,\prior^*)
    \end{align*}
    so for all $\outp$,  $\Pr(\mech^* \text{ outputs }\outp~|~\secret_1,\prior^*) = \Pr(\mech^* \text{ outputs }\outp~|~\secret_2,\prior^*)$ 
    and therefore the output probabilities do not depend on the secret. Hence, the posterior distribution of a secret is the same as the prior distribution, thus there is 0 assessed leakage.

    Next, we consider multiple runs. Note that $\mech^*$ always outputs a pair $(\data_1, \data_2)$ and one of those is the true input. Thus, if $\mech^*$ is run $k$ times and one of the outputs is different from the rest, then the secret is revealed: if the first components are always the same, the secret is $\secret_1$, and if the second components are always the same, the secret is $\secret_2$.

    So, without loss of generality, suppose the true secret is $\secret_1$. Then the outputs of sequential runs of $\mech^*$ can be denoted as $(\data, \data^{(1)}_2), (\data, \data^{(2)}_2), (\data, \data^{(3)}_2), \dots$. If any $\data^{(i)}_2\neq \data^{(1)}_2$ then the secret is revealed. Thus the probability the secret is revealed on the $k^{th}$ run (for $k>1$) is $\sum_{\data^\prime} \Pr(\data^\prime~|~\secret_2,\prior^*)\Pr(\data^\prime~|~\secret_2,\prior^*)^{k-2}(1-\Pr(\data^\prime~|~\secret_2,\prior^*))$. Its expected value is then:
    \begin{align*}
        \lefteqn{\sum_{k=2}^\infty \sum_{\data^\prime} k \Pr(\data^\prime~|~\secret_2,\prior^*)^{k-1}(1-\Pr(\data^\prime~|~\secret_2,\prior^*))}\\
        &= \sum_{\data^\prime}\left(-(1-\Pr(\data^\prime~|~\secret_2,\prior^*)) + \sum_{k=1}^\infty  k \Pr(\data^\prime~|~\secret_2,\prior^*)^{k-1}(1-\Pr(\data^\prime~|~\secret_2,\prior^*))\right)\\
        &= \sum_{\data^\prime}\left(-(1-\Pr(\data^\prime~|~\secret_2,\prior^*)) + \frac{1}{1-\Pr(\data^\prime~|~\secret_2,\prior^*)}\right)\\
        &= \sum_{\data^\prime} \frac{1-(1-\Pr(\data^\prime~|~\secret_2,\prior^*))^2}{1-\Pr(\data^\prime~|~\secret_2,\prior^*)}\\
        &= \sum_{\data^\prime} \frac{2\Pr(\data^\prime~|~\secret_2,\prior^*) - \Pr(\data^\prime~|~\secret_2,\prior^*)^2}{1-\Pr(\data^\prime~|~\secret_2,\prior^*)}\\
        &= \sum_{\data^\prime} \Pr(\data^\prime~|~\secret_2,\prior^*)\frac{2 - \Pr(\data^\prime~|~\secret_2,\prior^*)}{1-\Pr(\data^\prime~|~\secret_2,\prior^*)}\\
        &= \sum_{\data^\prime} \Pr(\data^\prime~|~\secret_2,\prior^*)\left(1 + \frac{1}{1-\Pr(\data^\prime~|~\secret_2,\prior^*)}\right)\\
        &= 1 + \sum_{\data^\prime} \Pr(\data^\prime~|~\secret_2,\prior^*)\left(\frac{1}{1-\Pr(\data^\prime~|~\secret_2,\prior^*)}\right)\\
            &\leq 1 + \sum_{\data^\prime} \Pr(\data^\prime~|~\secret_2,\prior^*)\left(\frac{1}{1-\max_\data \max_{\secret_i}\Pr(\data~|~\secret_i)}\right)\\
            &= 1 + \frac{1}{1-\max_\data \max_{\secret_i}\Pr(\data~|~\secret_i)}
    \end{align*}
\end{proofE}

\textbf{\emph{Multiple Secrets Example.}}
\emph{{Is the problem that there were too few secrets to defend?}} A plausible attempt to counter such privacy collapse is to add more secrets so that every bit in the dataset is covered by a potential secret. However, this is not always effective: 

\begin{example}\label{ex:two}
Consider a database of $n\geq 3$ bits, $\data=\{x_1,\dots, x_n\}$. Let the secrets $\secret^{(1)}_0,\dots, \secret^{(n)}_0$ and $\secret^{(1)}_1,\dots, \secret^{(n)}_1$ be defined as $\secret^{(i)}_j=$``$i$-th bit of the data is $j$''. Let $\priorset=\{\prior^*\}$, where the prior $\prior^*$ generates an $n$-bit dataset by flipping a fair coin $n$ times.  Define the mechanism $\mech^*(x_1,\dots,x_n)$ as the following ($\xor$ is the xor operation): 
\begin{itemize}[leftmargin=*]
    \item With probability $1/2$, output the tuple containing these $n-1$ values: $\allowbreak(x_2\xor x_1, ~~~~~~~~
    x_3\xor x_1, ~~~~~~~~\dots, ~~~~~~~~x_n\xor x_1)$. We call this output type A.
    \item Otherwise, if $n$ is odd, output $x_1\xor x_2\xor\cdots\xor x_n$ (the parity of the database), and if $n$ is even, output $x_2\xor\cdots\xor x_n$ (the parity of the database after omitting the first bit). We call this output type B.
\end{itemize}
\end{example}
If an attacker sees output type A, they can xor the $n-1$ values together. If $n$ is odd, this equals  $x_2\xor\cdots \xor x_n$. If $n$ is even, this equals $x_1\xor x_2\xor\cdots \xor x_n$. If the attacker later sees output type B, they can xor it with the output type A and learn the value of $x_1$.
Once $x_1$ is revealed, the attacker can use it with output type A to determine all the other secrets. The dataset could be reconstructed by running $\mech^*$ as few as 2 times, but when run only once, $\mech^*$ still satisfies zero leakage under posterior-based privacy, summarized in Theorem~\ref{the:discomposition_2}:

\begin{theoremE}\label{the:discomposition_2}
Let $\mech^*$ be the mechanism from Example \ref{ex:two}. Given an output $\outp$ of $\mech^*$,  $\Pr(\secret^{(i)}_j~|~\mech^*(\Data)=\outp,\prior^*)=\Pr(\secret^{(i)}_j~|~\prior^*)$ for all $i=1,\dots, n$ and $j\in\{0,1\}$. Thus, $\mech^*$ has zero leakage under posterior-based privacy definitions. 
However, the expected number of runs before the input dataset is revealed is $3$.
\end{theoremE}

\begin{proofE}
    There are two types of outputs (a tuple or a single bit), and we consider each in turn.
    First, we consider the case where the output is a single number. 
    It is easy to see that for all $i$:
    \begin{align*}
        1/4 &= \Pr(\outp=1 ~|~\secret^{(i)}_2,\prior) = \Pr(\outp=1 ~|~\secret^{(i)}_1,\prior)\\
        1/4 &= \Pr(\outp=0 ~|~\secret^{(i)}_2,\prior) = \Pr(\outp=0 ~|~\secret^{(i)}_1,\prior)
    \end{align*}
    since at least one database bit is uniformly random (because there are at least 3 bits in the database and hence at least 2 in the $\xor$ operation)  and it makes the xor appear uniformly random. 

    Next, we consider the case where the output is a tuple of (2 or more) elements, which we denote as $(y_2,\dots, y_n)$.
    \begin{align*}
      \lefteqn{   \Pr(\outp = (y_2,\dots, y_n)~|~\secret^{(1)}_1,\prior) }\\
      &= \frac{1}{2} \Pr(\data=(1, 1\xor y_2, \dots, 1 \xor y_n)) = 2^{-n+1}\\
      \lefteqn{   \Pr(\outp = (y_2,\dots, y_n)~|~\secret^{(1)}_0,\prior) }\\
      &= \frac{1}{2} \Pr(\data=(0, 0\xor y_2, \dots, 0 \xor y_n)) = 2^{-n+1}\\
      \intertext{Note that by symmetry, the cases of secret pair $(\secret^{(2)}_0, \secret^{(2)}_1)$ are identical to $(\secret^{(i)}_0, \secret^{(i)}_1)$ for $i>2$}
      \lefteqn{   \Pr(\outp = (y_2,\dots, y_n)~|~\secret^{(2)}_0,\prior) }\\
      &= \frac{1}{2} \Pr(\data=\Big((y_2 \xor 0), 0, (y_3 \xor y_2 \xor 0),\dots, (y_n \xor y_2 \xor 0))\Big) \\
      & = 2^{-n+1}\\
        \lefteqn{   \Pr(\outp = (y_2,\dots, y_n)~|~\secret^{(2)}_1,\prior) }\\
      &= \frac{1}{2} \Pr(\data=\Big((y_2 \xor 1), 1, (y_3 \xor y_2 \xor 1),\dots, (y_n \xor y_2 \xor 1))\Big) \\
      & = 2^{-n+1}\\
    \end{align*}
    and so $\mech^*$ satisfies $0$-posterior-based privacy since the output probabilities are not affected
    by any secret.
    
    We next consider the expected number of runs of $\mech^*$ before we see one of each type of output (a single number and a tuple). In the first run, we get one of them. Then the distribution of the amount of additional runs of   $\mech^*$ that are needed until the second type of output is seen is a geometric distribution with expected value of 2 (each run gives us a 1/2 probability of seeing that second type of output). Hence, the expected number of runs is 1+2=3.
\end{proofE}
\textbf{\emph{Multiple Priors Example.}}
Previous examples leveraged the prior distribution to reveal secrets. \emph{Would the composition properties improve when using more priors?} 
Specifically analyzing Pufferfish privacy, we again show that this is not the case. 

When $\priorset$ is the set of all product distributions, it is known that the resulting instantiation of Pufferfish will turn into $\epsilon$-differential privacy \cite{pufferfish}. However, this is typically not how posterior-based privacy definitions would be used. The ideal situation would be to specify a reference prior $\prior^*$ that is a good initial guess of the true distribution, and then to add to $\priorset$ many ``nearby'' priors to hedge against inaccuracies in $\prior^*$. In this way, the collection $\priorset$ of priors would encode approximate knowledge, such as ``the conditional probability that a cancer patient is female is between $0.35$ and $0.65$'', ``the conditional probability that a cancer-free patient is female is between $0.4$ and $0.6$'', etc.

Unfortunately, this kind of $\priorset$ causes problems with composition. We show a mechanism that has a finite Pufferfish privacy parameter when run once, but can reveal the secret when run twice.
\begin{example}\label{ex:three}
We simplify the discussion by considering one secret pair $(\secret_1,\neg \secret_1)$. The set $\priorset$ contains many priors, but in line with the above discussion, has the following structure: there exist two sets of databases $\universe_{T}$ and $\universe_F$ with the following properties:
\begin{enumerate}[leftmargin=*]
   \item Secret $\secret_1$ is true for all $\data\in \universe_T$ and $\secret_1$ is false for all $\data\in \universe_F$. 
    \item $\Pr(\data\in \universe_T~|~ \secret_1,\prior)\in [L_T, U_T]$ for all $\prior\in\priorset$ ($L_T>0$, $U_T<1$).
    \item $\Pr(\data\in \universe_F~|~ \neg\secret_1,\prior)\in [L_F, U_F]$ for all $\prior\in\priorset$ ($L_F>0$, $U_F<1$).
\end{enumerate}
Define $\mech^*$ that, on input $\data$, does the following:
\begin{itemize}[leftmargin=*]
    \item \textbf{Case 1.} If $\secret_1(\data)$ is true: (\textbf{1a}) if $\data\in \universe_T$,  return $(``1a", ``2a")$ with probability $0.5$ and $(``1a", ``2b")$ otherwise; (\textbf{1b}) if $\data\notin \universe_T$, return $(``1b", ``2a")$ with probability $0.5$ and $(``1b", ``2b")$ otherwise.
    \item \textbf{Case 2.} If $\secret_1(\data)$ is false: (\textbf{2a}) if $\data\in \universe_F$, return $(``1a", ``2a")$ with probability $0.5$ and $(``1b", ``2a")$ otherwise; (\textbf{2b}) if $\data\notin \universe_F$, return $(``1a", ``2b")$ with probability $0.5$ and $(``1b", ``2b")$ otherwise.
\end{itemize}
\end{example}
$\mech^*$ from Example \ref{ex:three} has 4 total cases (1a, 1b, 2a, 2b), and each case returns a tuple of case numbers, one of which is the actual case number. 
When $\mech^*$ is run multiple times, as soon as two different outputs are observed, the true case, hence the secret, is revealed (the tuple element that didn't change is the real case number).
The following theorem shows that $\mech^*$ has a finite privacy parameter under Pufferfish, yet the expected number of runs before the secret is revealed is 3, meaning that composition breaks down quickly.

\begin{theoremE}\label{the:multi-prior-puff}
    Let $\allowbreak\epsilon=\max(\log|\frac{U_T}{L_F}|,\log|\frac{L_T}{U_F}|,\log|\frac{1-U_T}{U_F}|, \log|\frac{1-L_T}{L_{F}}|,\\\log|\frac{L_T}{1-L_F}|, \log|\frac{U_T}{1-U_{F}}|, \log|\frac{1-U_T}{1-L_F}|, \log|\frac{1-L_T}{1-U_{F}}|)$. Mechanism $\mech^*$ from Example \ref{ex:three} satisfies $\epsilon$-Pufferfish privacy. The expected number of runs of $\mech^*$ to reveal the true secret is 3. 
\end{theoremE}

\begin{proofE}
    For all four possible outputs from the mechanism $\mech$, the information leakage is bounded by $\epsilon$ in the following ways:
\begin{align*}
\lefteqn{    \frac{\Pr(\text{output is } = (``1a", ``2a") ~|~ \secret_1,\prior)}{\Pr(\text{output is } = (``1a", ``2a") ~|~ \neg\secret_1,\prior)}}\\
&= \frac{\sum_{\data\in S_{T}} \Pr(\data~|~\secret_1,\prior) \frac{1}{2}}{\sum_{\data\in S_{F}} \Pr(\data~|~\neg\secret_1,\prior) \frac{1}{2}}\in [\frac{L_T}{U_F}, \frac{U_T}{L_{F}}] 
\end{align*}
\begin{align*}
\lefteqn{    \frac{\Pr(\text{output is } = (``1b", ``2a") ~|~ \secret_1,\prior)}{\Pr(\text{output is } = (``1b", ``2a") ~|~ \neg\secret_1,\prior)}}\\
&= \frac{\sum_{\data\notin S_{T}} \Pr(\data~|~\secret_1,\prior) \frac{1}{2}}{\sum_{\data\in S_{F}} \Pr(\data~|~\neg\secret_1,\prior) \frac{1}{2}}\in [\frac{1-U_T}{U_F}, \frac{1-L_T}{L_{F}}]
\end{align*}
\begin{align*}
\lefteqn{    \frac{\Pr(\text{output is } = (``1a", ``2b") ~|~ \secret_1,\prior)}{\Pr(\text{output is } = (``1a", ``2b") ~|~ \neg\secret_1,\prior)}}\\
&= \frac{\sum_{\data\in S_{T}} \Pr(\data~|~\secret_1,\prior) \frac{1}{2}}{\sum_{\data\notin S_{F}} \Pr(\data~|~\neg\secret_1,\prior) \frac{1}{2}}\in [\frac{L_T}{1-L_F}, \frac{U_T}{1-U_{F}}] 
\end{align*}
\begin{align*}
\lefteqn{    \frac{\Pr(\text{output is } = (``1b", ``2b") ~|~ \secret_1,\prior)}{\Pr(\text{output is } = (``1b", ``2b") ~|~ \neg\secret_1,\prior)}}\\
&= \frac{\sum_{\data\notin S_{T}} \Pr(\data~|~\secret_1,\prior) \frac{1}{2}}{\sum_{\data\notin S_{F}} \Pr(\data~|~\secret_2,\prior) \frac{1}{2}}\in [\frac{1-U_T}{1-L_F}, \frac{1-L_T}{1-U_{F}}] 
\end{align*}

We next show that the expected number of runs of $\mech^*$ to reveal the true secret is 3. Given a dataset $\data$, $\mech$ can only produce two types of output, each with probability 0.5. Once the mechanism $\mech$ provides different outputs from multiple runs, we are able to infer whether the true secret is $\secret_1$ or $\neg\secret_1$. On the first run, we observe one type. After that, we keep running until we observe the other type. Each additional run has a probability of 0.5 of giving that other type, so the number of expected extra runs is a geometric distribution with a mean of 2. Therefore, the total expected runs = 1 + 2 = 3.
\end{proofE}
Again, when all possible priors are added in $\priorset$, Pufferfish becomes $\epsilon$-differential privacy, and hence composes linearly.
However, the above example implies that such a composition would only happen when enough priors are added, such that for any set $S$ of databases with the same secret,
$\priorset$ must include priors where $S$ has arbitrarily low probabilities and also include priors where $S$ has arbitrarily high probabilities.
While such a case may improve composition, it will ruin utility as such a comprehensive set of priors will include priors that the data curator views as unrealistic.

In the next section, we study necessary conditions for Pufferfish privacy to compose, and show that it needs differential-privacy-like constraints on mechanisms. 
Greatly expanding the set of priors will be an overkill, i.e.,  expanding the set of priors until composition is achieved would be equivalent to adding the necessary conditions that we found as well as many others (unnecessarily).

\section{Necessary Conditions for Composition and Post-processing Invariance}\label{sec:conditions}
The previous section demonstrated that posterior-based privacy definitions, including Pufferfish, can
suffer privacy collapse.
In this
section, we study how to avoid this situation. Specifically, suppose we have a collection of mechanisms
$\mech_1,\dots, \mech_k$ where each $\mech_i$ satisfies Pufferfish with privacy parameter $\epsPuffer_i$.
What conditions do these mechanisms need to satisfy in order to linearly compose with each other, so that
the combined release of all of their outputs satisfies Pufferfish with privacy parameter $\leq \sum_{i=1}^k\epsPuffer_i$?

The necessary conditions that we prove (Theorem \ref{thm:nfc}) identify a class of constraints that we call the NfC (\underline{n}ecessary-\underline{f}or-\underline{c}omposition) constraints. For each dataset $\data$, one or more NfC constraints need to be added (in addition to the Pufferfish constraints); otherwise, linear composition will fail. As we will show, the differential privacy constraints, $\log \Pr(\mech(\data)=\outp)\leq \epsilon + \log \Pr(\mech(\data^\prime)=\outp)$, are special cases of NfC constraints, which have the general form of: $\log \Pr(\mech(\data)=\outp)\leq \epsilon + \sum_{i=1}^m\beta_i\log \Pr(\mech(\data_i)=\outp)$.

Then, we study another important property of privacy definitions called \emph{post-processing invariance}. Post-processing invariance means that if $\mech_i$ satisfies $\epsPuffer_i$-Pufferfish and $f_i$ is a possibly randomized algorithm, then the algorithm $f_i\circ\mech_i$, which releases $f_i(\mech_i(\data))$, still satisfies Pufferfish with privacy parameter $\epsPuffer_i$. Now that extra NfC constraints must be added, one needs to ensure that those constraints are also post-processing invariant. That is, if $\mech$ satisfies the extra constraints, then $f_i\circ\mech_i$ should also satisfy the constraints. Or, put another way, if $\mech_1,\dots,\mech_k$ compose linearly, then so should $f_1\circ\mech_1, \dots, f_k\circ\mech_k$. 

Our second result (Theorem \ref{thm:post-process-gen}) states that a non-redundant subset of  NfC constraints is post-processing invariant if and only if every constraint in that subset is a differential-privacy-style constraint.
Together, these results imply that if one wants to design Pufferfish mechanisms that compose with each other, those mechanisms should also satisfy a form of differential privacy. In Section \ref{sec:puffer-mech}, we show how to design such algorithms for tabular datasets (i.e., sufficient conditions for composition).

\subsection{Necessary Condition for Composition}\label{sec:nfc}
We present necessary conditions for a set of Pufferfish mechanisms to compose, when defending a set of pairs of potential secrets $\secretset$ against priors in  $\priorset$. The conditions use the concept of a \emph{convex combination vector} $\beta$, which is a vector whose entries are nonnegative and sum up to 1. Entry $\ell$ in the vector is represented as $\beta_\ell$.

One condition in the following theorem is that the set of possible databases is finite (but arbitrarily large). In practice, this is not a restrictive condition because attribute domains are typically bounded (e.g., a 500-character limit for URLs, 64 bits for numeric values) and only finitely many records can be collected in a finite time (e.g., no database will have more than $2^{64}$ records). This simplifies the proof as it avoids dealing with corner cases that only arise with infinite-dimensional vectors.
Similarly, the necessary conditions consider a set of mechanisms $\mech_1,\dots, \mech_m$, whose ranges are finite. Mechanisms with infinite output spaces would also need to satisfy these and possibly additional conditions.

\begin{theoremE}[Necessary Condition for Composition (NfC)] \label{thm:nfc}
    Let $\priorset$ be a set of priors over a finite (but arbitrarily large) set of possible datasets.
    Let  $\secretset$ be a set of secret pairs, and let $0 \le epsilon<\infty$ be the Pufferfish privacy parameter. Let $\mech_1,\dots, \mech_m$ be a set of mechanisms, with finite ranges, that satisfy $\epsilon$-Pufferfish and linearly compose with each other. Then for all ordered pairs $(\secret_1,\secret_2)\in S$ and datasets $\data$ for which $\secret_1$ is true, there exists a convex combination vector $\beta$ over datasets $\data_\ell$ for which $\secret_2$ is true, such that:
    \begin{align}
        \log \Pr(\mech(\data)=\outp) \leq \epsilon + \hspace{-1em}\sum_{\data_\ell: \secret_2(\data_\ell)\text{ is true}}\beta_\ell\log \Pr(\mech(\data_\ell)=\outp),\label{eq:thenfcconstraints}
    \end{align}
    for each mechanism and each of its possible outputs $\outp$. 
    Note the convex combination depends on $\secret_1, \secret_2$ and $\data$, and the reverse secret pair $(\secret_2,\secret_1)$ may have its own convex combination.
\end{theoremE}
\begin{proofE}
    Let $\universe=\{\data_1,\data_2,\dots, \data_\numdata\}$ be a finite collection of datasets.
Given an instantiation of Pufferfish with the set of priors $\priorset$ and set of secret pairs $\secretset$, let $\mechset_\epsilon=\{\mech_1,\dots, \mech_m\}$ be a set of mechanisms with a finite output space that satisfy $\epsilon$-Pufferfish and compose linearly together. 
Every output $\outp$ of a mechanism $\mech$ has a corresponding $\numdata$-dimensional likelihood vector $\LL(\outp, \mech)$ defined as:
$$\LL(\outp, \mech)=[\Pr(\mech(\data_1)=\outp),~ \dots, ~\Pr(\mech(\data_\numdata)=\outp)].$$
Let $\lvset_\epsilon$ be the collection of all likelihood vectors of mechanisms in $\mechset_\epsilon$. That is,
$$\lvset_\epsilon = \{\LL(\outp,\mech)~|~\mech\in\mechset_\epsilon \text{ and }\outp\in\range(\mech)\}.$$
Now, pick $k$ vectors $\lv_1,\dots, \lv_k$ from $\lvset_\epsilon$ (e.g., all of them), so that each $\lv_i$ corresponds to an $\LL(\outp,\mech)$ of some $\mech\in\mechset_\epsilon$ and $\outp\in\range(\mech)$.
\begin{figure}
    \centering
    \includegraphics[width=0.5\linewidth]{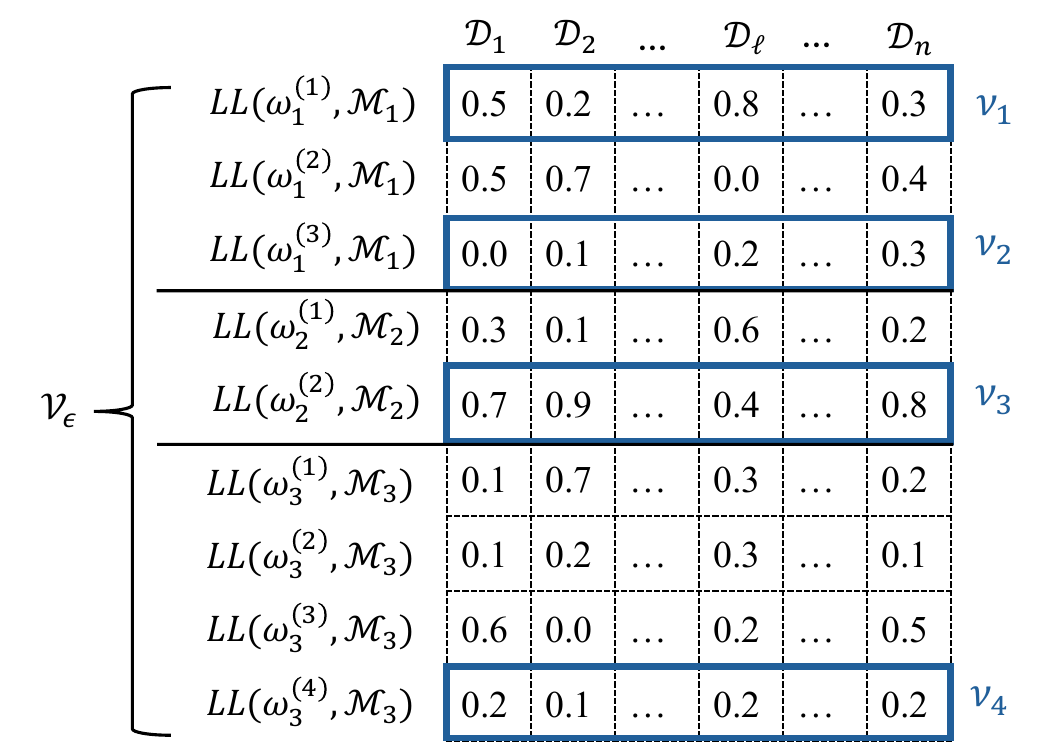}
    \caption{Example illustration with $\mechset_\epsilon=\mech_1, \mech_2, \mech_3$, with each mechanism having 3, 2, and 4 unique outputs. In this illustration, $k=4$ vectors are picked.}
    \label{fig:sec4notation}
\end{figure}

Figure~\ref{fig:sec4notation} illustrates an example of $\lvset_\epsilon$ and $\lv_i$ vectors. In this example, $\lvset_\epsilon$ consists of three mechanisms, $\mech_1, \mech_2, \mech_3$, and each mechanism has three ($\outp^{(1)}_1$, $\outp^{(2)}_1$, $\outp^{(3)}_1$), two ($\outp^{(1)}_2$, $\outp^{(2)}_2$), and four ($\outp^{(1)}_3$,\dots, $\outp^{(4)}_3$) possible outputs, respectively.
The chosen $k=4$ vectors are $\lv_1=\LL(\outp^{(1)}_1, \mech_1)$, $\lv_2=\LL(\outp^{(3)}_1, \mech_1)$, $\lv_3=\LL(\outp^{(2)}_2, \mech_2)$, and $\lv_4=\LL(\outp^{(4)}_3, \mech_3)$.

Let $\lv_i[\ell]$ denote the $\ell$-th entry (column) of $\lv_i$, e.g., $\lv_1[2]=0.2$, $\lv_3[1]=0.7$, etc.
For each $\lv_i=\LL(\outp,\mech)$, suppose the mechanism $\mech$ was applied to observe the output $\outp$ for $c_i$ times, where $c_i$ is a nonnegative integer.
By linear composition, we must have:
\begin{align*}
    \sum_{i=1}^k c_i \epsilon 
    &\geq \log \frac{
        \sum_{\ell=1}^\numdata \Pr_\prior(\data_\ell ~|~ \secret_1)\prod\limits_{i=1}^k \lv_{i}[\ell]^{c_i}
    }{
        \sum_{\ell=1}^\numdata \Pr_\prior(\data_\ell ~| \secret_2)\prod\limits_{i=1}^k \lv_{i}[\ell]^{c_i}
    },
\end{align*}
for each $\prior\in\priorset$ and secret pair $(\secret_1,\secret_2)\in \secretset$.
For any positive integer $r$, by multiplying each $c_i$ by $r$ gives:
\begin{align*}
    \epsilon 
    &\geq \log \frac{
      \left(  \sum_{\ell=1}^\numdata \Pr_\prior(\data_\ell ~|~ \secret_1)\prod\limits_{i=1}^k \lv_{i}[\ell]^{rc_i} \right)^{1/\sum_j rc_j}
    }{
       \left( \sum_{\ell=1}^\numdata \Pr_\prior(\data_\ell ~| \secret_2)\prod\limits_{i=1}^k \lv_{i}[\ell]^{rc_i}\right)^{1/\sum_j rc_j}
    }\\
    &=\log\frac{
    \left(\sum_{\ell=1}^\numdata \Pr_\prior(\data_\ell ~|~ \secret_1)\prod\limits_{i=1}^k\left(\lv_{i}[\ell]^{c_i/\sum_j c_j}\right)^{\sum_j rc_j}\right)^{1/\sum_j rc_j}
    }{
    \left(\sum_{\ell=1}^\numdata \Pr_\prior(\data_\ell ~|~ \secret_2)\prod\limits_{i=1}^k\left(\lv_{i}[\ell]^{c_i/\sum_j c_j}\right)^{\sum_j rc_j}\right)^{1/\sum_j rc_j}
    }.\\
\end{align*}
The numerator and the denominator are both a weighted $L_p$ norm with $p=r\sum_j c_j$ and weights $\Pr_\prior(\data_\ell ~|~ \secret)\in [0, 1]$. As $p\rightarrow\infty$, both converge to the max among the items 
that have nonnegative weight. Thus, letting $r\rightarrow\infty$, we have:
\begin{align*}
    \epsilon 
    &\geq\log\frac{
    \left(\max\limits_{\{\ell: \data_\ell\text{ has }\secret_1\}}\prod\limits_{i=1}^k\left(\lv_{i}[\ell]^{c_i/\sum_j c_j}\right)\right)
    }{
    \left(\max\limits_{\{\ell: \data_\ell\text{ has }\secret_2\}}\prod\limits_{i=1}^k\left(\lv_{i}[\ell]^{c_i/\sum_j c_j}\right)\right)
    }.\\
\end{align*}
Next, we make the observation that any collection of nonnegative rational numbers $q_1,\dots, q_k$ that add up to 1 can be represented as $\frac{c_1}{\sum_j c_j}, \dots, \frac{c_k}{\sum_j c_j}$ for an appropriate choice of $c_1,\dots, c_k$. Since the rationals are dense in the real numbers, and due to continuity, for any $\alpha_1,\dots, \alpha_k$ that are nonnegative and add up to 1, we have:
\begin{align}
    \epsilon 
    &\geq\log\frac{
    \left(\max\limits_{\{\ell: \data_\ell\text{ has }\secret_1\}}\prod\limits_{i=1}^k\left(\lv_{i}[\ell]^{\alpha_i}\right)\right)
    }{
    \left(\max\limits_{\{\ell: \data_\ell\text{ has }\secret_2\}}\prod\limits_{i=1}^k\left(\lv_{i}[\ell]^{\alpha_i}\right)\right)
    }\nonumber\\
    &=
    \left(\max\limits_{\{\ell: \data_\ell\text{ has }\secret_1\}}\sum\limits_{i=1}^k {\alpha_i}\log(\lv_{i}[\ell])\right)
    -
    \left(\max\limits_{\{\ell: \data_\ell\text{ has }\secret_2\}}\sum\limits_{i=1}^k{\alpha_i}\log(\lv_{i}[\ell])\right).
    \label{eq:comp_necessary}
\end{align}
which is the necessary condition for linear composition to hold.
Note that the condition should hold 
for all secret pairs $(\secret_1,\secret_2)\in\secretset$ (and a corresponding  equation for $(\secret_2,\secret_1)$), for all integers $k>0$, for all $\lv_1,\dots, \lv_k\in \lvset_\epsilon$, and for all convex combination coefficients $\alpha_1,\dots, \alpha_k$.

\vspace{1em}
\textbf{The Necessary Condition as a Collection of Linear Programs}\\
\vspace{1em}

To better understand the necessary condition in Equation~\ref{eq:comp_necessary}, we turn it into a collection of linear programs. As Equation~\ref{eq:comp_necessary} cannot directly be expressed as a linear program, we find a collection of linear programs such that Equation \ref{eq:comp_necessary} is satisfied if and only if all the linear programs have an optimal value $\leq \epsilon$.

Given choices of a secret pair $(\secret_1,\secret_2)\in\secretset$, an integer $k>0$, and $k$ vectors $\lv_1,\dots, \lv_k\in \lvset_\epsilon$, let $\Xi_1 = \{\ell~:~ \secret(\data_\ell)= \secret_1\}$ be the indexes for datasets whose secret is $\secret_1$ and $\Xi_2$ is the set of datasets whose secret is $\secret_2$.
First, we observe that Equation \ref{eq:comp_necessary} can be turned into $|\Xi_1|$ equations:
\begin{align*}
    &\text{ for all }\ell^*\in \Xi_1\\
    &\text{ for all convex combination coefficients $\alpha_1,\dots,\alpha_k$}\\
    \epsilon 
    &\geq
    \left(\sum\limits_{i=1}^k {\alpha_i}\log(\lv_{i}[\ell^*])\right)
    -
    \left(\max\limits_{\ell\in\Xi_2}\sum\limits_{i=1}^k{\alpha_i}\log(\lv_{i}[\ell])\right)\\
    &= \min_{\ell\in\xi_2} \sum\limits_{i=1}^k{\alpha_i} \left(\log(\lv_{i}[\ell^*])-\log(\lv_{i}[\ell])\right).
\end{align*}
Clearly, Equation \ref{eq:comp_necessary} is true only if all of the above equations are true.
We convert each of the above equations into a separate linear program by introducing a new optimization variable, $\epsilon_0$, which measures the size of the right-hand side following the $\geq$, and treating $\alpha_1,\dots, \alpha_k$ as linear program variables.
The resulting set of linear optimization problems is:
\begin{align*}
    & \text{ for each }\ell^*\in\Xi_1:\\
    \max_{\epsilon_0, \alpha_1,\dots, \alpha_k} &\epsilon_0\\
    \text{s.t. }& \alpha_1 \geq 0, \dots, \alpha_k \geq 0\\
                & \alpha_1 + \cdots + \alpha_k = 1\\
                & \left(\sum\limits_{i=1}^k {\alpha_i}\log(\lv_{i}[\ell^*]/\lv_{i}[\ell])\right) - \epsilon_0 \geq 0\text{ for all }\ell\in\Xi_2.
\end{align*}
Equation \ref{eq:comp_necessary} is satisfied if and only if each of these linear programs has a solution $\leq \epsilon$.

\vspace{1em}
\textbf{Finding the Dual Linear Programs}\\
\vspace{1em}

The dual linear programs \cite{borwein2006convex} of the above linear programs are:
\begin{align*}
        & \text{ for each }\ell^*\in\Xi_1:\\
    \min_{x, ~\beta_\ell \text{ for }\ell\in\Xi_2} \quad & x \\
    \text{s.t. } \quad& \beta_\ell \leq 0 \text{ for }\ell\in\Xi_2\\
                 & -\sum_{\ell\in \Xi_2}\beta_\ell =1\\
                 & x + \sum_{\ell\in \Xi_2}\beta_\ell\log(\lv_{i}[\ell^*]/\lv_{i}[\ell])\geq 0 \text{ for }i=1,\dots, k,
\end{align*}
which can be simplified to:
\begin{align*}
        & \text{ for each }\ell^*\in\Xi_1:\\
    \min_{x, ~\beta_\ell \text{ for }\ell\in\Xi_2} \quad & x \\
    \text{s.t. } \quad & \beta_\ell \geq 0 \text{ for }\ell\in\Xi_2\\
                 & \sum_{\ell\in \Xi_2}\beta_\ell =1\\
                 & x - \log(\lv_{i}[\ell^*]) + \sum_{\ell\in \Xi_2}\beta_\ell\log(\lv_{i}[\ell])\geq 0 \text{ for }i=1,\dots, k. 
\end{align*}
Strong duality theorem \cite{borwein2006convex} shows that the optimal value of the dual program equals the optimal value of the original linear programs (since the original programs always have a solution).
Hence, the sufficient condition in Equation \ref{eq:comp_necessary} is satisfied if and only if the optimal solution to all of these dual problems is $\leq \epsilon$.
In order to ensure that each dual program has a solution $\leq \epsilon$, we need to show that
for each $j\in\Xi_1$ there is a collection of coefficients $\beta^{(j)}_\ell$ for $\ell\in \Xi_2$ such that: 
\begin{itemize}
\item the $\beta^{(j)}$ are nonnegative and $\sum_{\ell\in\Xi_2} \beta^{(j)}_\ell=1$ (i.e., $\beta^{(j)}_\ell$ for $\ell\in\Xi_2$ are convex combination coefficients). 
\item $\epsilon \geq  \log(\lv_{i}[\ell^*]) - \sum_{\ell\in \Xi_2}\beta^{(j)}_\ell\log(\lv_{i}[\ell])$ for $i=1,\dots, k$.
\end{itemize}
Rewriting these with the language of mechanisms gives Theorem~\ref{thm:nfc}.
\end{proofE}

We use the term \textbf{NfC constraints} to refer to all inequalities having the form shown in Equation \ref{eq:thenfcconstraints}. Before presenting the proof sketch of Theorem \ref{thm:nfc},  note that Equation \ref{eq:thenfcconstraints} can also be written as:
\begin{align*}
    \Pr(\mech(\data)=\outp) \leq e^\epsilon \prod\limits_{\data_\ell: \secret_2(\data_\ell)\text{ is true}} \Pr(\mech(\data_\ell)=\outp)^{\beta_\ell}.
\end{align*}
This makes the connection to differential privacy clearer: if $\beta_\ell=1$ for some index $\ell$ and 0 for the rest, we get $\Pr(\mech(\data)=\outp) \leq e^\epsilon  \Pr(\mech(\data_\ell)=\outp)$, where $\secret_1$ is true of $\data$ and $\secret_2$ is true of $\data_\ell$. Hence, differential privacy inequalities are a subset of NfC constraints.

\begin{proof}[proof sketch (see \shortlong{\cite{nfcarxiv}}{the appendix} for the full proof).]~
   The overall idea is that we convert the statement ``linear composition holds'' into a collection of linear programs and take the duals of the linear programs.
    %
    Since the number of possible outputs is finite, we can write them as $\outp_1,\dots, \outp_k$. Without loss of generality, we can assume that the ranges of the mechanisms are disjoint, i.e., given an output $\outp_i$, we can tell which mechanism it came from. Therefore, we can use the shorthand $P(\outp_i~|~\data_\ell)$ to mean  $P(\mech_j(\data_\ell)=\outp_i)$, where $\mech_j$ is the mechanism whose range contains $\outp_i$.
    
    The data curator selects $r$ mechanisms with replacement from $\{\mech_1,\dots, \mech_m\}$ (mechanisms can be repeated), runs them, and records their outputs. Suppose that the fraction of times that output $\outp_i$ occurs is $\alpha_i$. Under linear composition, the Pufferfish privacy cost would be bounded by $r\epsilon$, and for every $\prior\in\priorset$ and secret pair $(\secret_1,\secret_2)$, we get the equation:
    \begin{align*}
        r\epsilon &\geq  \log \frac{
      \left(  \sum_{\ell} \Pr_\prior(\data_\ell ~|~ \secret_1)\prod\limits_{i=1}^k \Pr(\outp_i~|~\data_\ell)^{r\alpha_i} \right)
    }{
       \left( \sum_{\ell} \Pr_\prior(\data_\ell ~|~ \secret_2)\prod\limits_{i=1}^k \Pr(\outp_i~|~\data_\ell)^{r\alpha_i}\right).
    }
    \end{align*}
    If you divide by $r$ and let $r\rightarrow\infty$, this becomes:
    \begin{align*}
   \epsilon &\geq  \left(\max\limits_{\{\ell: \secret_1(\data_\ell)=true\}}\sum\limits_{i=1}^k {\alpha_i}\log(\Pr(\outp_i~|~\data_\ell))\right)\\
   &\phantom{\geq}  -
    \left(\max\limits_{\{\ell: \secret_2(\data_\ell)=true\}}\sum\limits_{i=1}^k{\alpha_i}\log(\Pr(\outp_i~|~\data_\ell))\right).
    \end{align*}
    Such an equation must hold for any $(\secret_1,\secret_2)\in\secretset$ and any choice of nonnegative $\alpha_1,\dots, \alpha_k$ that add up to 1.
    Then, for every secret pair $(\secret_1,\secret_2)$ and each $\data_{\ell^*}$ for which $\secret_1$ is true, we create a linear program as follows. We show that the above constraints are true if and only if all of the following linear programs have a solution $\leq \epsilon$. 
    \begin{align*}
    &\hspace{-3em}\text{For each $(\secret_1,\secret_2)\in\secretset$ and $\data_{\ell*}$ for which $\secret_1(\data_{\ell^*})$=true:}\\
    \max_{\epsilon_0, \alpha_1,\dots, \alpha_k} &\epsilon_0\\
    \text{s.t. } &\alpha_1 \geq 0, \dots, \alpha_k \geq 0\\
                 &\alpha_1 + \cdots + \alpha_k = 1\\
                 &\text{For all }\data_\ell \text{ having }\secret_2(\data_\ell)=true:\\
                 &\left(\sum\limits_{i=1}^k {\alpha_i}\log\left(\frac{\Pr(\outp_i~|~\data_{\ell^*})}{\Pr(\outp_i~|~\data_\ell)}\right)\right) - \epsilon_0 \geq 0.
   \end{align*}
    We convert each such linear program into the corresponding dual program for each $(\secret_1,\secret_2)\in\secretset$ and $\data_{\ell*}$ for which $\secret_1(\data_{\ell^*})$=true:
    \begin{align*}
    &\min_{\epsilon^\prime_0, ~ \text{ convex combination vector } \beta}  \quad  \epsilon^\prime_0 \\
    &\text{s.t. }  \epsilon^\prime_0 - \log(\Pr(\outp_i~|~\data_{\ell^*})) + \sum_{\ell}\beta_\ell\log(\Pr(\outp_i~|~\data_{\ell}))\geq 0 \text{ for all }\outp_i, 
   \end{align*}
   where the summation is over all $\ell$ for which $\secret_2(\data_\ell)=true$. We conclude that, for composition,
   all of these dual programs must have a solution $\leq \epsilon$. If that is true, each of these programs will come up with a linear combination vector $\beta$, which becomes Equation \ref{eq:thenfcconstraints}.

\end{proof}

\subsection{Enforcing Post-processing Invariance}
Theorem~\ref{thm:nfc} tells us that if we want composition properties inside Pufferfish, we need to add additional NfC constraints (inequalities with the form of Equation \ref{eq:thenfcconstraints}). However, such constraints are not necessarily post-processing invariant. That is, after post-processing, a mechanism may fail to satisfy the constraints, and the post-processed mechanisms may not compose. We next show that if a set of NfC constraints is ``non-redundant'' and post-processing invariant, then each one of those constraints must have the form of: $\log \Pr(\mech(\data_i)=\outp)\leq \epsilon + \log \Pr(\mech(\data_j)=\outp)$, i.e., each one of those constraints must be a differential privacy-style constraint. 

To understand redundancy, consider the following two constraints with two different convex combination vectors $\beta^{(0)}$, $\beta^{(1)}$:
\begin{align*}
 \log \Pr(\mech(\data^*)=\outp) &\leq \epsilon + \sum_{\data_\ell: \secret_2(\data_\ell)\text{ is true}}\beta^{(0)}_\ell\log \Pr(\mech(\data_\ell)=\outp),\\
  \log \Pr(\mech(\data^*)=\outp) &\leq \epsilon + \sum_{\data_\ell: \secret_2(\data_\ell)\text{ is true}}\beta^{(1)}_\ell\log \Pr(\mech(\data_\ell)=\outp).
\end{align*}
Then, for any $c\in(0,1)$, one can define $\beta^{(c)}=c\beta^{(1)} + (1-c)\beta^{(0)}$ to get the following redundant constraint: 
\begin{align*}
      \log \Pr(\mech(\data^*)=\outp) &\leq \epsilon + \sum_{\data_\ell: \secret_2(\data_\ell)\text{ is true}}\beta^{(c)}_\ell\log \Pr(\mech(\data_\ell)=\outp).
\end{align*}
In general, a constraint is redundant if it is implied by the other constraints that we have. By iteratively removing redundant constraints,
we obtain a non-redundant set of constraints.

\begin{theoremE}[NfC and Post-processing Invariance]
\label{thm:post-process-gen}
    Suppose there are finitely many datasets and $m$ non-redundant NfC constraints, where constraint $i$ involves some dataset $\data_i$, some secret pair $(\secret_1^{(i)}, \secret_2^{(i)})$, and a convex combination vector $\beta^{(i)}$, and has the form: 
    $$\log \Pr(\mech(\data_i)=\outp) \leq \epsilon + \sum_{\data_\ell: \secret^{(i)}_2(\data_\ell)\text{ is true}}\beta^{(i)}_\ell\log \Pr(\mech(\data_\ell)=\outp),$$
    for all $\outp$.
    If all of the convex combination vectors have only 1 nonzero entry (i.e., these are all DP-style conditions), this set of constraints is post-processing invariant. If at least one of the convex combination vectors has 2 or more nonzero entries, then this set of constraints is not post-processing invariant.
\end{theoremE}
\begin{proofE}
    Clearly, if each $\beta$ is a one-hot encoding (one entry contains a 1 and all others are 0), the resulting constraints are postprocessing invariant. 

    Next, let $m$ be the number of constraints. Let $\beta^{(i)}$ denote the convex combination used in the $i^{\text{th}}$ constraint and $\data_{t_i}$ be the corresponding dataset used on the left hand side (i.e., $t_i$ is the index of the dataset among $\{\data_1,\dots,\data_n\}$) and $(\secret_1^{(i)}, \secret_2^{(i)})$ be the associated secret pair. Note that each convex combination $\beta^{(i)}$ can be turned into a vector $w^{(i)}$ that has a component for every possible dataset and 
    \begin{align*}
        \log \Pr(\mech(\data_{t_i})=\outp) - \sum\limits_{\data_\ell: \secret_2^{(i)}(\data_\ell)\text{ is true}}\beta^{(i)}_\ell\log \Pr(\mech(\data_\ell)=\outp)\\
       =  \sum_{j=1}^n w^{(i)}_j\log\Pr(\mech(\data_j)=\outp)
    \end{align*}
    Furthermore, it is clear that for any $c$, 
    \begin{align*}
    \sum_{j=1}^n c w^{(i)}_j = 0\quad
    \text{since }\quad c - \sum_{\data_\ell~: ~\secret_2(\data_\ell)\text{ is true}} c\beta^{(i)}_\ell = 0.
    \end{align*}

    Next, define two sets of vectors:
    \begin{align*}
        \Omega_1 &=\{v~:~ v\cdot w^{(i)}\leq \epsilon \text{ for }i=1,\dots, m\}\\
        \Omega_2 &= \{v~: \text{each component of $v$ is $\leq 0$} \}
    \end{align*}
    Then clearly a mechanism $\mech$ satisfies those $m$ constraints if and only if, for every output $\outp$, the $n$-dimensional vector $$[\log(\Pr(\mech(\data_1))=\outp), \dots, \log(\Pr(\mech(\data_n))=\outp)]\in \Omega_1\cap\Omega_2$$.

    Next, without loss of generality, let $\beta^{(1)}$ be the convex combination that has at least 2 nonzero components. Our proof strategy takes the following steps:
    \begin{itemize}
        \item First we show there exist two linearly independent vectors $v^\prime \in \Omega_1\cap\Omega_2$ and $v^*\in\Omega_1\cap\Omega_2$  such that $v^\prime\cdot w^{(1)}=v^* \cdot w^{(1)}=\epsilon$.
        \item Then we show that there exists a mechanism $\widehat{\mech}$ whose possible outputs are 0,1,2,3, that satisfies necessary constraints for composition, $\Pr(\widehat{\mech}(\data_i)=0)=e^{v^\prime_i}$ (for $i=1,\dots, n$) and $\Pr(\widehat{\mech}(\data_i)=1)=e^{v^*_i}$ (for $i=1,\dots, n$).
        \item There exists a randomized algorithm $\randalg$ such that the postprocessed algorithm $\widehat{\mech}^*(\data)\equiv\randalg(\widehat{\mech}(\data))$ does not satisfy the necessary constraints for composition.
    \end{itemize}

\noindent\textbf{Step 1:}
Note that the constant vectors, those whose entries are $[c,\dots, c]$ for $c < 0$,  satisfy all the necessary-for-composition constraints with strict inequality (i.e., $<$) and belong to the interior of $\Omega_1\cap\Omega_2$.
Thus, there is a point $v^{(b)}\in\Omega_2$  for which $v^{(b)} \cdot w{(i)} < \epsilon$ for $i=1,\dots, m$.

Next, we note that the constraint involving $\beta^{(1)}$ and hence $w^{(1)}$ is not redundant. Therefore, there exists a vector $v^{(a)}\in\Omega_2$ for which $v^{(a)}\cdot w^{(1)} > \epsilon$ but $v^{(a)} \cdot w{(i)} \leq \epsilon$ for $i=2,\dots, m$ and without loss of generality, we may assume that $v^{(a)} \cdot w{(i)} < \epsilon$ for $i=2,\dots, m$
by continuity of the dot product (e.g., by moving the vector slightly in the direction of a point that is in the interior of $\Omega_1 \cap\Omega_2$).

Thus, by convexity of the halfspaces and compactness in any bounded subset, the line from $v^{(a)}$ to $v^{(b)}$  contains some point $v^{(\dagger)}$ for which $v^{(\dagger)}\cdot w^{(1)} = \epsilon$ and $v^{(\dagger)} \cdot w{(i)} < \epsilon$ for $i=2,\dots, m$. Furthermore, due to compactness, one can create an appropriate open set around this line which allows us to conclude that $\{v\in\Omega_2:~|~v\cdot w^{(1)}=\epsilon\}\cap \left(\bigcap_{i=2}^m \{v\in\Omega_2~:~v\cdot w^{(i)}\}<\epsilon\right)$
has a nonempty relative interior. We note that this relative interior has dimensionality $n-1$ since the ambient space has dimensionality $n$. Furthermore, since $\beta^{(1)}$ has two nonzero components, $w^{(1)}$ has at least 3 nonzero components. Thus, one can make the following choices:
\begin{itemize}
    \item Choose two indexes $\gamma_1$ and $\gamma_2$ for which $w^{(1)}_{\gamma_1}\neq 0$ and $w^{(1)}_{\gamma_2}\neq 0$. This is possible because $w^{(1)}$ has at least 3 nonzero components.
    \item Two linearly independent vectors $v^*$ and $v^\prime$ in the relative interior of $\{v\in\Omega_2:~|~v\cdot w^{(1)}=\epsilon\}\cap \left(\bigcap_{i=2}^m \{v\in\Omega_2~:~v\cdot w^{(i)}\}<\epsilon\right)$ such that $v^*_{\gamma_1}-v'_{\gamma_1}\neq v^*_{\gamma_2}- v'_{\gamma_2}$. This is possible because the relative interior of that set has dimensionality $n-1\geq 2$. 
\end{itemize}

\textbf{Step 2:} Next we note that for any $c\geq 0$, the vector $v^{*}-[c, \dots, c]\in \Omega_1\cap\Omega_2$ and $v'-[c,\dots,c] \in\Omega_1\cap\Omega_2$. Furthermore, $[\log(1/2), \dots, \log(1/2)]$ is in the interior of $\Omega_1\cap\Omega_2$. We also note that:
\begin{align*}
\lim\limits_{c\rightarrow\infty}    \log\left([1/2, \dots, 1/2] - e^{v^*-[c,\dots, c]}\right)=\log([1/2,\dots,1/2])
\end{align*}
where the exponentiation and log are taken pointwise. Thus for some $c^*$, we must have:
\begin{align*}
u^{(0)}   &\equiv v^*-[c^*,\dots, c^*] \in \Omega_1\cap\Omega_2\\
u^{(1)} & \equiv \log\left([1/2, \dots, 1/2] - e^{v^*-[c^*,\dots, c^*]}\right) \in \Omega_1\cap \Omega_2\\
u^{(2)}  &\equiv v'-[c^*,\dots, c^*] \in \Omega_1\cap\Omega_2\\
u^{(3)}   &\equiv \log\left([1/2, \dots, 1/2] - e^{v'-[c^*,\dots, c^*]}\right) \in \Omega_1\cap \Omega_2\\
\end{align*}
and also: $e^{u^{(0)}} + e^{u^{(1)}} = [1/2, \dots, 1/2] = e^{u^{(2)}}+e^{u^{(3)}}$.

Thus, because of this fact, we can use those vectors to define the mechanism $\widehat{\mech}$ such that:
\begin{align*}
    \Pr(\widehat{\mech}(\data_i) = 0) &= e^{u^{(0)}_i}\\
        \Pr(\widehat{\mech}(\data_i) = 1) &= e^{u^{(1)}_i}\\
            \Pr(\widehat{\mech}(\data_i) = 2) &= e^{u^{(2)}_i}\\
                \Pr(\widehat{\mech}(\data_i) = 3) &= e^{u^{(3)}_i}
\end{align*}
and $\widehat{\mech}$ satisfies the necessary-for-composition constraints because $u^{(0)}\in\Omega_1\cap \Omega_2$ and $u^{(1)}\in\Omega_1\cap \Omega_2$ and $u^{(2)}\in\Omega_1\cap \Omega_2$ and $u^{(3)}\in\Omega_1\cap \Omega_2$.

We next note that for outputs $\outp=0$ and $\outp=2$, the necessary-for-composition constraint associated with $\beta^{(1)}$ is tight for $\widehat{\mech}$. Recalling that $\data_{t_1}$ is the dataset that appears on the left-hand side for the constraint involving $\beta^{(1)}$ and the corresponding secret pair is $(\secret^{(1)}_1, \secret^{(1)}_2)$,
\begin{align*}
    &\log\Pr(\widehat{\mech}(\data_{t_1})=0) - \sum\limits_{\data_\ell~:~ \secret_2^{(1)}(\data_\ell)\text{ is true}}\beta^{(1)}_\ell\log \Pr(\widehat{\mech}(\data_\ell)=0)
    \\
    &=\sum_{i=1}^n w^{(1)}_i \log\Pr(\widehat{\mech}(\data_i)=0)\\
    &=\sum_{i=1}^n w^{(1)}_i u^{(0)}_i\\
    &=\sum_{i=1}^n w^{(1)}_i (v^{*}_i-c^*)\\
    &=\sum_{i=1}^n w^{(1)}_i v^{*}_i=\epsilon\\
\end{align*}
where the second-to-last equality is because of the construction of $w^{(1)}$ (i.e., $\sum_i w^{(1)}_i c=0$ for all $c$) and the last equality is due to the construction of $v^*$.

A similar calculation for $\outp=2$ and $v^\prime$ shows that 
$$    \log\Pr(\widehat{\mech}(\data_{t_1})=2) - \sum\limits_{\data_\ell~:~ \secret_2^{(1)}(\data_\ell)\text{ is true}}\beta^{(1)}_\ell\log \Pr(\widehat{\mech}(\data_\ell)=2) = \epsilon
    $$

\noindent\textbf{Step 3:} Next we construct a postprocessing algorithm $\randalg$ such that the $\randalg\circ\widehat{\mech}$ (the algorithm that runs $\randalg$ on the output of $\widehat{\mech}$) does not satisfy the necessary-for-composition constraints.

Again, recall that $\data_{t_1}$ is the dataset that appears on the left-hand side for the constraint involving $\beta^{(1)}$ and the corresponding secret pair is $(\secret^{(1)}_1, \secret^{(1)}_2)$.
Pick a $p^*>0$ such that $p^* \leq \min(\Pr(\widehat{\mech}(\data_{t_1})=0), ~\Pr(\widehat{\mech}(\data_{t_1})=2))$.
Define the postprocessing algorithm $\randalg$ as:
\begin{align*}
    \randalg(0) &= \begin{cases} 1 &\text{ with probability }\frac{p^*}{\Pr(\widehat{\mech}(\data_{t_1})=0)}\\0 & \text{ otherwise }\end{cases}\\
    \randalg(1) &= 0\\
    \randalg(2) &= \begin{cases} 1 &\text{ with probability }\frac{p^*}{\Pr(\widehat{\mech}(\data_{t_1})=2)}\\0 & \text{ otherwise }\end{cases}\\
    \randalg(3) &= 0
\end{align*}
Then:
\begin{align}
     & \log\Pr(\randalg(\widehat{\mech}(\data_{t_1}))=1) - \hspace{-0.5cm}\sum\limits_{\data_\ell~:~ \secret_2^{(1)}(\data_\ell)\text{ is true}}\hspace{-0.5cm}\beta^{(1)}_\ell\log \Pr(\randalg(\widehat{\mech}(\data_\ell))=1)\nonumber\\
     &= \sum\limits_{i=1}^n w^{(1)}_i\log \Pr(\randalg(\widehat{\mech}(\data_i))=1)\nonumber\\
     &= \sum\limits_{i=1}^n w^{(1)}_i\log \left(p^*\frac{\Pr(\widehat{\mech}(\data_i)=0)}{\Pr(\widehat{\mech}(\data_{t_1})=0)}+p^*\frac{\Pr(\widehat{\mech}(\data_i)=2)}{\Pr(\widehat{\mech}(\data_{t_1})=2)}\right)\nonumber\\
     &= \sum\limits_{i=1}^n w^{(1)}_i\log \left(\frac{\Pr(\widehat{\mech}(\data_i)=0)}{\Pr(\widehat{\mech}(\data_{t_1})=0)}+\frac{\Pr(\widehat{\mech}(\data_i)=2)}{\Pr(\widehat{\mech}(\data_{t_1})=2)}\right)\nonumber\\
     \intertext{(since $\sum_i w_i p^*=0$)}
     &\geq \sum\limits_{i=1}^n w^{(1)}_i\log\Bigg(2\sqrt{\frac{\Pr(\widehat{\mech}(\data_i)=0)}{\Pr(\widehat{\mech}(\data_{t_1})=0)}\frac{\Pr(\widehat{\mech}(\data_i)=2)}{\Pr(\widehat{\mech}(\data_{t_1})=2)}}\Bigg)\label{eq:amgm}
     \intertext{(by the arithmetic mean/geometric mean inequality, with equality only when $\frac{\Pr(\widehat{\mech}(\data_i)=0)}{\Pr(\widehat{\mech}(\data_{t_1})=0)}=\frac{\Pr(\widehat{\mech}(\data_i)=2) }{\Pr(\widehat{\mech}(\data_{t_1})=2)}$ for all $i$ where $w^{(1)}_i\neq 0$)}\nonumber
     &= \sum\limits_{i=1}^n w^{(1)}_i\log\Bigg(\sqrt{\frac{\Pr(\widehat{\mech}(\data_i)=0)}{\Pr(\widehat{\mech}(\data_{t_1})=0)}\frac{\Pr(\widehat{\mech}(\data_i)=2)}{\Pr(\widehat{\mech}(\data_{t_1})=2)}}\Bigg)\nonumber\\
     &= \frac{1}{2}\sum\limits_{i=1}^n w^{(1)}_i\log\Bigg(\frac{\Pr(\widehat{\mech}(\data_i)=0)}{\Pr(\widehat{\mech}(\data_{t_1})=0)}\Bigg)
      +\frac{1}{2}\sum\limits_{i=1}^n w^{(1)}_i\log\Bigg({\frac{\Pr(\widehat{\mech}(\data_i)=2)}{\Pr(\widehat{\mech}(\data_{t_1})=2)}}\Bigg)\nonumber\\
    &= \frac{1}{2}\sum\limits_{i=1}^n w^{(1)}_i\log\Bigg({\Pr(\widehat{\mech}(\data_i)=0)}\Bigg) +\frac{1}{2}\sum\limits_{i=1}^n w^{(1)}_i\log\Bigg({{\Pr(\widehat{\mech}(\data_i)=2)}}\Bigg)\nonumber
    \intertext{(since $\sum_i w_i c=0$ for any constant $c$)}
     &=\epsilon/2 + \epsilon/2 = \epsilon\nonumber
\end{align}
This sequence of derivations has an inequality $\leq$ that is only true when $\frac{\Pr(\widehat{\mech}(\data_i)=0)}{\Pr(\widehat{\mech}(\data_{t_1})=0)}=\frac{\Pr(\widehat{\mech}(\data_i)=2) }{\Pr(\widehat{\mech}(\data_{t_1})=2)}$ for all $i$ where $w^{(1)}_i\neq 0$.

However, 
\begin{align*}
    &\frac{\Pr(\widehat{\mech}(\data_i)=0)}{\Pr(\widehat{\mech}(\data_{t_1})=0)}=\frac{\Pr(\widehat{\mech}(\data_i)=2) }{\Pr(\widehat{\mech}(\data_{t_1})=2)}\text{ for all $i$ where $w^{(1)}_i\neq 0 $}\\
    &\Leftrightarrow 
    \log\frac{\Pr(\widehat{\mech}(\data_i)=0)}{\Pr(\widehat{\mech}(\data_{t_1})=0)}=\log\frac{\Pr(\widehat{\mech}(\data_i)=2) }{\Pr(\widehat{\mech}(\data_{t_1})=2)}\text{ for all $i$ where $w^{(1)}_i\neq 0 $}\\
     &\Leftrightarrow 
    v^*_i-v^*_{t_1}=v'_i-v'_{t_1}\text{ for all $i$ where $w^{(1)}_i\neq 0 $}\\
    &\Leftrightarrow v^*_i-v'_i = v^*_j-v'_j \text{ for all $i,j$ where $w^{(1)}_i\neq 0 $ and $w^{(1)}_j\neq 0$}
\end{align*}
However, this is impossible by the construction of $v^*$ and $v'$ at the end of Step 1. Therefore, Equation \ref{eq:amgm} must be a strict inequality, and that means that $\randalg\widehat{\mech}$ does not satisfy at least one of the necessary-for-composition constraints.
\end{proofE}

\begin{proof}[proof sketch (see \shortlong{\cite{nfcarxiv}}{the appendix} for the full proof).]
The proof is based on the observations (1) that the feasible region (allowable real-valued vectors of the form $[\Pr(\mech(\data_1)=\outp),\dots, \Pr(\mech(\data_{\numdata})=\outp)]$) for the system of constraints is only convex when each $\beta^{(i)}$ has a single non-zero component and (2) post-processing acts like a convex operation on those allowable vectors. Hence, we construct a mechanism for which the constraints are tight but are violated after postprocessing.
\end{proof}

The conclusion we reach is that, if we want a set of Pufferfish mechanisms to compose with each other and be post-processing invariant, we should design those mechanisms to satisfy differential privacy. 
The next section presents sufficient conditions---how to construct those differential privacy mechanisms and choose their privacy parameter $\epsDP$, so that they also satisfy Pufferfish with a desired privacy parameter $\epsPuffer$ and compose linearly.
Our results are specific to $\epsilon$-Pufferfish (Definition \ref{def:pufferfish}), which mirrors $\epsilon$-differential privacy (Definition \ref{def:dp}).
For other variations of Pufferfish that are analogous to $(\epsilon, \delta)$-differential privacy or Renyi-DP~\cite{pierquin2024renyi}, we conjecture that similar results would hold; e.g., composable $(\epsilon, \delta)$-Pufferfish must also satisfy $(\epsilon, \delta)$-DP.

\section{Constructing General Pufferfish Mechanisms with Composition
}\label{sec:puffer-mech}
Theorem~\ref{thm:nfc} and Theorem~\ref{thm:post-process-gen} show that a set of linearly composable Pufferfish mechanisms that support post-processing invariance must satisfy a set of constraints that
whittle down
to the types that are used by differential privacy.
These insights provide formal support for prior work \cite{song2017composition,song2017pufferfish,cao2017quantifying} that designed linearly composable Pufferfish mechanisms by taking $\epsDP$-differential privacy mechanisms and showing that they also satisfy $\epsPuffer$-Pufferfish for some $\epsPuffer \geq \epsDP$. Specifically, Song et al. \cite{song2017pufferfish,song2017composition} proposed an extension of the Laplace Mechanism (called MQM) to defend against attacker priors that all share the same Markov network structure. Meanwhile, Cao et al. \cite{cao2017quantifying} extend arbitrary DP mechanisms to Markov chains and defend against finitely many Markov chain priors while allowing the Pufferfish privacy parameter to compose.

In this section, we generalize these ideas to show that in any Pufferfish setting for tabular data (including finite Markov chains), the same recipe holds. Any collection of per-entry $\epsDP$-differential privacy mechanisms can be translated into $\epsPuffer$-Pufferfish privacy, and the Pufferfish parameters linearly compose. Furthermore, the translation between $\epsDP$ and $\epsPuffer$ is guided by a new concept we call an $\ab$-influence curve. The $\ab$-influence curve only depends on the set of secret pairs $\secretset$ and the priors in $\priorset$. Although the $\ab$-influence curve can be computationally expensive to compute exactly, one only needs an upper bound to get privacy guarantees. The tighter the upper bound is, the better the resulting utility. This is analogous to sensitivity computation in differential privacy, which is easy to upper bound (providing privacy guarantees), but can be computationally difficult to estimate exactly. 

Before going into the technical details, we first explain the intuition behind the $\ab$-influence curve. Suppose there is an infectious disease dataset and a person $A$ in this dataset. Person $A$ has a set $\highinf_1$ containing $b_1$ people who are highly likely to get infected if $A$ gets sick (e.g., $\highinf_1$ is the set of close family members), a set $\highinf_2$ (with $\highinf_1\subset \highinf_2$) containing $b_2$ people who are moderately likely to get infected if $A$ gets sick (e.g., $\highinf_2$ contains $\highinf_1$ and close friends), and a third set $\highinf_3$ containing $b_3$ people who are only slightly likely to get infected from $A$ (e.g., $\highinf_3$ contains $\highinf_2$ and acquaintances). 

Now, suppose $A$ can increase the chance of infections of people not in $\highinf_1$ by at most a factor of $a_1$, can increase the chance of infection of people not in $\highinf_2$ (resp., $\highinf_3$) by at most a factor of $a_2 <a_1$ (resp. $a_3 < a_2$), and so on. These tuples $(a_1, b_1), (a_2, b_2), (a_3, b_3)$ form the $\ab$-influence curve. That is, we can think of the $a$ value as a function of $b$. 
When considering $b$ number of people most closely associated with $A$ (or the number of cells in a tabular dataset most closely associated with a secret), 
$a(b)$ upper bounds the effect of $A$ on  anyone outside of this close group of size $b$. 

This $\ab$-influence curve provides translations between $\epsDP$ and $\epsPuffer$ as follows. An algorithm satisfying $\epsDP$-differential privacy also satisfies Pufferfish with privacy parameter $\epsPuffer=b\epsDP + a(b)$ for any choice of $b$. This means that the best $\epsPuffer$ parameter is $\min_b b\epsDP + a(b)$. Thus, if one has the $\ab$-influence curve, or knows at least a few points on the curve, they can compute the upper bound on the Pufferfish privacy parameter.

In order to show how the $\ab$-influence curve is a generalization of the work of Song et al. ~\cite{song2017pufferfish}, we briefly describe the Markov Quilt Mechanism and \emph{max-influence} concept in Section \ref{sec:mqminf}. Then, we formally define $\ab$-influence, prove the translation between $\epsDP$ and $\epsPuffer$, and present the slightly-better-than-additive composition results in Section \ref{sec:sufficient}. 
We discuss practical considerations in estimating the $\ab$-curve, including robustness to mis-specifying priors in Section \ref{sec:practical}.
We also give worked out examples of the $\ab$-influence curve in Section \ref{sec:abcurve}.

\subsection{The MQM and Max-Influence~\cite{song2017pufferfish}}\label{sec:mqminf}

Song et al.~\cite{song2017pufferfish} showed that applying the Laplace mechanism of differential privacy to a database generated by a Markovian structure can realize composable Pufferfish privacy.
As we will show later, this finding is a special instance of our generalized result.
We briefly describe the result from~\cite{song2017pufferfish} as a background.

Song et al.~\cite{song2017pufferfish} focused on a specific case where each entry of the database has a dependency described by a Markovian-structure graph $G$.
Figure~\ref{fig:markov-quilt} shows an example where the database $X=\{X_1,\dots, X_n\}$ is generated by a Markov chain, i.e., the value of $X_i$ only depends on $X_{i-1}$. 
They introduced a concept called a \emph{Markov Quilt}, which is an extension of the Markov Blanket~\cite{markov_blanket}. A set of nodes $X_Q$ is a Markov Quilt for a node $X_i$ if: (1) conditioning on $X_Q$, the graph $G$ is partitioned into two independent node sets,  $X_N$ and $X_R$,  such that $X=X_N \cup X_Q \cup X_R$, and $X_i \in X_N$, and (2) for all $x_R \in \mathcal{X}^{card(X_R)},  x_Q \in \mathcal{X}^{card(X_Q)},  x_i \in \mathcal{X}$,  we have $\Pr[X_R=x_R|X_Q=x_Q,  X_i=x_i] = \Pr[X_R=x_R|X_Q=x_Q]$.
In other words, the graph is partitioned in a way that the values of nodes $X_R$ only depend on $X_Q$, and not on $X_N$ (see Figure~\ref{fig:markov-quilt}).

Song et al.~\cite{song2017pufferfish} proposed a composable Pufferfish privacy mechanism called a Markov Quilt Mechanism (MQM) that can hide the value of each node $X_i$ in the Markov chain. 
In fact, the MQM mechanism with $\epsPuffer$-Pufferfish privacy is equivalent to a Laplace mechanism of differential privacy of a certain $\epsDP$ (i.e., a Laplace noise of scale $\frac{L}{\epsDP}$ is added to an $L$-Lipschitz query output).
The relationship between $\epsPuffer$ and $\epsDP$ is as follows.
First, they define \emph{max-influence} as below:
\begin{definition}[Max-Influence~\cite{song2017pufferfish}]\label{def:max-influence-markov}
    The max-influence of a variable $X_i$ on a set of variables $X_A$ under $\priorset$ is defined as:
    \begin{align*}
        e_{\priorset}(X_A|X_i) = \sup_{\prior \in \priorset} \max_{a,b \in \mathcal{X}, x_A \in \mathcal{X}^{card(X_A)}}\log \Bigg | \frac{\Pr(X_A=x_A|X_i=a, \prior)}{\Pr(X_A=x_A|X_i=b, \prior)}\Bigg |. 
    \end{align*}
\end{definition}
Then, Song et al.~\cite{song2017pufferfish} showed that applying a Laplace mechanism with $\epsDP$-differential privacy in fact achieves composable $\epsPuffer$-Pufferfish privacy, where 
$\epsDP=\frac{\epsPuffer-e_{\priorset}(X_Q|X_i)}{card(X_N)}$. Here, $card(X_N)$ is the cardinality of $X_N$.
MQM only works when there exists a partition such that $\epsPuffer > e_{\priorset}(X_Q|X_i)$, and the proposal~\cite{song2017pufferfish} was to enumerate all the possible partitions and choose the partition with the largest $\epsDP$ for utility.

\begin{figure}
    \centering
    \includegraphics[width=1.\linewidth]{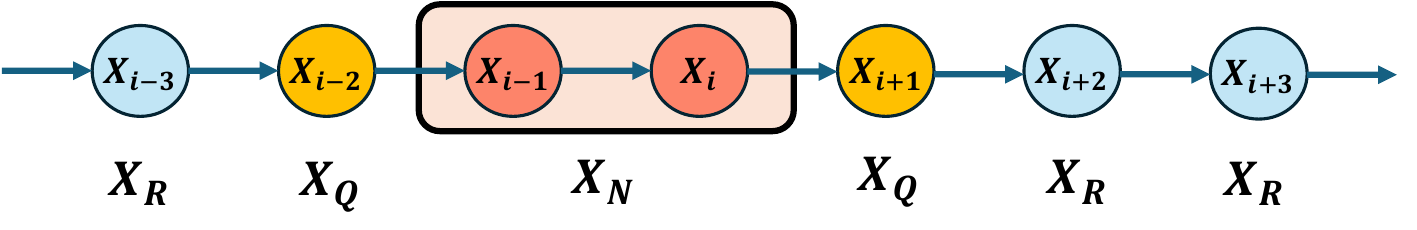}
    \caption{An illustration of $X_R, X_Q, X_N$ in Markov Chain. 
    }
    \label{fig:markov-quilt}
\end{figure}

\subsection{Sufficient Condition for Composition}
\label{sec:sufficient}

The MQM mechanism~\cite{song2017pufferfish} only works under databases whose generation follows a Markovian structure and is restricted to Laplace noise.
In this section, we introduce a generalized framework that can translate an arbitrary per-entry $\epsDP$-differential privacy mechanism into a composable Pufferfish mechanism.
We consider a general tabular dataset with $m$ records and $n$ attributes without loss of generality. Let  $I=\{(i,j)~:~i\in\{1,\dots, m\}, j\in\{1,\dots, n\}\}$ be the entry index set of the entire dataset ($|I|=mn$). 
We do not restrict the generation of the dataset (i.e., it doesn't have to come from a Markov chain, etc.).

Assume an arbitrary index set $ \lowinf \subseteq I$, which we consider to be \emph{low-influence} region, and a complement set $\highinf = I \setminus \lowinf$, which we consider as \emph{high-influence} region.
$\dataquilt$ is the entry values of $\data$ restricted to indices in $\lowinf$, and the corresponding random variable is $\Dataquilt$.
While the separation of $\highinf$ and $\lowinf$ can be arbitrary, we are interested in finding a separation where the information leakage about the secret is minimal through $\lowinf$, when the size of $\highinf$ is restricted to a certain value, $b$.
The worst-case leakage through $\lowinf$ in such optimal ($\highinf$, $\lowinf$) pairs for all possible $b$ values can be summarized in a single curve, which we call the $\ab$-influence curve.
We formally define \emph{$\ab$-influence} below:

\begin{definition}[$\ab$-influence]\label{def:ab-influence}
    Assume an arbitrary split of the index set $I$, $\lowinf \subseteq I$ and $\highinf = I\setminus\lowinf$.
    Given any $(\secret_i,\secret_j)\in \secretset$ and any possible $\Dataquilt$,
    let $\universe_{\lowinf}=\{\dataquilt_1,\dataquilt_2,\dots\}$ denote the set of all such
    shared subsets supported by $\priorset$.
    The worst-case information leakage through $\lowinf$ for secret pair $(\secret_i, \secret_j)$ is,
    \begin{align*}
    a^{\dagger}_{(\secret_i, \secret_j)}(\lowinf)
    \;=\;
    \sup_{\prior\in\priorset}
    \max_{\dataquilt\in\universe_{\lowinf}}\Bigg |
    \log\frac{
    \Pr[\Dataquilt=\dataquilt\mid \secret_i,\prior]
    }{
    \Pr[\Dataquilt=\dataquilt\mid \secret_j,\prior]
    }\Bigg |.
    \end{align*}
    If we define $a^\star_{(\secret_i,\secret_j)}(b)$ for $(\secret_i, \secret_j)$ as:
    \begin{align*}
        a^\star_{(\secret_i,\secret_j)}(b) \;=\; \inf_{\lowinf \subseteq I:\, |\highinf|\le b} a^{\dagger}_{(\secret_i,\secret_j)}(\lowinf),
    \end{align*}
    the $\ab$-influence curve is an upper-bounding curve for all $a^\star_{(\secret_i,\secret_j)}(b)$:
    \begin{align*}
    a(b) \;\ge\; \max_{(\secret_i,\secret_j)\in\secretset} a^{\star}_{(\secret_i,\secret_j)}(b),
    \end{align*}
    and we say that $(\priorset,\secretset)$ satisfies \emph{$\ab$-influence}. Note that by definition, the curve is monotonically non-increasing.
\end{definition}

Essentially, the $\ab$-influence curve is a succinct summary of the risk of revealing a counterfactual database, where $\lowinf$ entries stay the same, but at most $b$ number of $\highinf$ entries are replaced.

\begin{theoremE}[Translation between DP and Pufferfish]\label{the:markov-trick}
Fix a family of priors $\priorset$, and a set of secret pairs $\secretset$. 
Suppose $(\priorset,\secretset)$ satisfies $\ab$-influence (Definition~\ref{def:ab-influence}), and $(a,b)$ is a point on the $\ab$-influence curve.
Suppose a randomized mechanism $\mech$ achieves per-entry $\epsDP$-differential privacy, with $\epsDP \leq \frac{\epsPuffer - a}{b}$ and $a < \epsPuffer$. Then, $\mech$ satisfies $\epsPuffer$-Pufferfish privacy with respect to the secret set $\secretset$ and set of priors $\priorset$.
\end{theoremE}

\begin{proofE}
    For any $\outp$, we randomly select $ |\dataquilt| = |I| - b$ entries from the dataset and conditioned them as $\Dataquilt = \dataquilt$, denote the rest $b_1$ entries as $\dataneighbor$, 
\begin{align*}
   & \frac{\Pr(\mech(\Data)=\outp|\secret_i)}{\Pr(\mech(\Data)=\outp|\secret_j)}\\
   = & \frac{\sum_{\dataquilt} \Pr(\mech(\Data)=\outp|\secret_i, \Dataquilt=\dataquilt)\Pr(\Dataquilt=\dataquilt|\secret_i)}{\sum_{\dataquilt} \Pr(\mech(\Data)=\outp|\secret_j, \Dataquilt=\dataquilt)\Pr(\Dataquilt=\dataquilt|\secret_j)} \\
    \leq & \max_{\dataquilt}  \frac{\Pr(\mech(\Data)=\outp|\secret_i, \Dataquilt=\dataquilt)}{\Pr(\mech(\Data)=\outp|\secret_j, \Dataquilt=\dataquilt)} \cdot \frac{\Pr(\Dataquilt=\dataquilt|\secret_i)}{\Pr(\Dataquilt=\dataquilt|\secret_j)},
\end{align*}

The first ratio is bounded by $\epsPuffer-a$ in the following way:
{\small
\begin{align*}
& \frac{\Pr(\mech(\Data)=\outp|\secret_i, \Dataquilt=\dataquilt)}{\Pr(\mech(\Data)=\outp|\secret_j, \Dataquilt=\dataquilt)}\\
= &  \frac{\sum_{\dataneighbor} \Pr(\mech(\Data)=\outp|\secret_i, \Data=[\dataquilt, \dataneighbor])\Pr(\dataneighbor|\secret_i, \dataquilt)}{\sum_{\dataneighbor'}\Pr(\mech(\Data)=\outp|\secret_j, \Data=[\dataquilt, \dataneighbor'])\Pr(\dataneighbor'|\secret_j, \dataquilt)}\\
\leq &  \frac{\left(\max_{\dataneighbor} \Pr(\mech(\Data)=\outp|\secret_i, \Data=[\dataquilt, \dataneighbor])\right)\sum_{\dataneighbor}\Pr(\dataneighbor|\secret_i, \dataquilt)}{\left(\min_{\dataneighbor'}\Pr(\mech(\Data)=\outp|\secret_j,\Data=[\dataquilt,\dataneighbor'])\right)\sum_{\dataneighbor'}\Pr(\dataneighbor'|\secret_j,\dataquilt)}\\
= & \frac{\max_{\dataneighbor} \Pr(\mech(\Data)=\outp|\secret_i, \Data=[\dataquilt, \dataneighbor])}{\min_{\dataneighbor'}\Pr(\mech(\Data)=\outp|\secret_j, \Data=[\dataquilt, \dataneighbor'])}\\
\leq& \exp(\epsPuffer-a),
\end{align*}}
By $\epsDP$-differential privacy and group privacy, for any $\dataneighbor,\dataneighbor'$ the ratio
$\frac{\Pr(\mech(\Data)=\outp|\secret_i, \Data=[\dataquilt, \dataneighbor])}{\Pr(\mech(\Data)=\outp|\secret_j, \Data=[\dataquilt, \dataneighbor'])}$
is at most $\exp(b\epsDP)$. 
The second ratio is bounded by the max influence $e_\priorset(\Dataquilt|\secret_i, \secret_j)$, which is bounded by $a$. Therefore, the total leakage is bounded by $\exp(\epsPuffer)$.  
\end{proofE}

Theorem \ref{the:markov-trick} can be written as $\epsPuffer=b\epsDP + a$, which can also be interpreted as saying that the Pufferfish parameter uses differential privacy's group privacy with a group size of $b$ (accounting for the $b\epsDP$ part of the formula) plus a penalty $a$ for cross-record correlation that extends beyond the group.
Our next theorem, about composition, shows that this extra penalty only needs to be paid once, so that the composition of the Pufferfish parameter is still linear but is better than just adding up the Pufferfish parameters:

\begin{theoremE}[Sufficient Condition for Composition]\label{the:sfc}
    Fix a family of priors $\priorset$ and a set of secret pairs $\secretset$. 
    Consider $k$ mechanisms $\mech_1, ..., \mech_k$ applied sequentially to the same dataset, each satisfying $\epsDP_1, ..., \epsDP_k$-differential privacy, where $\epsDP_\ell \leq \frac{\epsPuffer_\ell - a_\ell}{b_\ell}$.
    Here, $(a_\ell, b_\ell)$ are arbitrary points on the $\ab$-influence curve of $(\priorset,\secretset)$ with $a_\ell < \epsPuffer_\ell$.
    That is, each $\mech_\ell$ achieves $\epsPuffer_\ell$-Pufferfish according to Theorem~\ref{the:markov-trick}.
    Then, the joint release of the outputs 
    satisfies $(\max_\ell a_\ell +\sum_{\ell=1}^k \epsPuffer_\ell - \sum_{\ell=1}^k a_\ell)$-Pufferfish privacy.
    Note that the composition is sub-additive due to ($\max_\ell a_\ell - \sum_{\ell=1}^k a_\ell)$ being always negative.
\end{theoremE}

\begin{proofE}
    (The proof below is adapted from Appendix B in \cite{song2017composition}.)
    
    For any secret pair $(\secret_i, \secret_j)$, let the low-influence sets used to determine the $\epsDP$ parameters for $\mech_1, \dots, \mech_k$ be $\Dataquilt_1, \dots, \Dataquilt_k$, respectively. Without loss of generality, suppose $b_1 \leq b_\ell$ for $\ell=1,\dots, k$. Then we have the following:
\begin{align}
    &\frac{\Pr[\bigwedge_{\ell=1}^k\mech_\ell(\Data)=\outp_\ell|\secret_i]}{\Pr[\bigwedge_{\ell=1}^k\mech_\ell(\Data)=\outp_\ell|\secret_j]}\nonumber\\
    \leq \max_{\dataquilt_1} & \frac{\Pr[\bigwedge_{\ell=1}^k\mech_\ell(\Data)=\outp_\ell|\secret_i, \Dataquilt_1=\dataquilt_1]}{\Pr[\bigwedge_{\ell=1}^k\mech_\ell(\Data)=\outp_\ell|\secret_j, \Dataquilt_1=\dataquilt_1]} ~ 
     \frac{\Pr[\Dataquilt_1=\dataquilt_1 |\secret_i]}{\Pr[\Dataquilt_1=\dataquilt_1 |\secret_j]},\label{eqn:2nd_ratio_composition}
\end{align}
where we randomly select $|\lowinf_1| = |I| - b_1$ entries from the dataset and conditioned them as $\Dataquilt_1 = \dataquilt_1$. Denote the rest $b_1$ entries as $\dataneighbor$.  The first ratio is bounded by the following:
{ \small
\begin{align*}
    \frac{\Pr[\bigwedge_{\ell=1}^k\mech_\ell(\Data)=\outp_\ell|\secret_i, \Dataquilt_1=\dataquilt_1]}{\Pr[\bigwedge_{\ell=1}^k\mech_\ell(\Data)=\outp_\ell|\secret_j, \Dataquilt_1=\dataquilt_1]}
    &=\frac{\sum_{\dataneighbor}\Pr[\bigwedge_{\ell=1}^k\mech_\ell(\Data)=\outp_\ell|\Data=[\dataquilt_1, \dataneighbor]]~\Pr[\dataneighbor|\secret_i, \dataquilt_1]}{\sum_{\dataneighbor}\Pr[\bigwedge_{\ell=1}^k\mech_\ell(\Data)=\outp_\ell|\Data=[\dataquilt_1, \dataneighbor]]~\Pr[\dataneighbor|\secret_j,\dataquilt_1]}\\
    \leq &\frac{\left(\max_{\dataneighbor}\Pr[\bigwedge_{\ell=1}^k\mech_\ell(\Data)=\outp_\ell|\Data=[\dataquilt_1, \dataneighbor]]\right)~\sum_{\dataneighbor}\Pr[\dataneighbor|\secret_i, \dataquilt_1]}{\left(\min_{\dataneighbor}\Pr[\bigwedge_{\ell=1}^k\mech_\ell(\Data)=\outp_\ell|\Data=[\dataquilt_1, \dataneighbor]]\right)~\sum_{\dataneighbor}\Pr[\dataneighbor|\secret_j, \dataquilt_1]}\\
    =&\frac{\max_{\dataneighbor}\Pr[\bigwedge_{\ell=1}^k\mech_\ell(\Data)=\outp_\ell|\Data=[\dataquilt_1, \dataneighbor]]}{\min_{\dataneighbor}\Pr[\bigwedge_{\ell=1}^k\mech_\ell(\Data)=\outp_\ell|\Data=[\dataquilt_1, \dataneighbor]]} \\
    =& \prod_{\ell=1}^k \frac{\max_{\dataneighbor}\Pr[\mech_\ell(\Data)=\outp_\ell|\Data=[\dataquilt_1, \dataneighbor]]}{\min_{\dataneighbor}\Pr[\mech_\ell(\Data)=\outp_\ell|\Data=[\dataquilt_1, \dataneighbor]]} 
    \leq \prod_{\ell=1}^k e^{b_1\epsDP_\ell} \\
    &\text{(using group privacy because up to $b_1$ values can change)}\\
    \leq & \exp\left(\sum_{\ell=1}^k \frac{\epsilon_\ell - a_\ell}{b_\ell}b_1\right)\\
    \leq & \exp\Big(\sum_{\ell=1}^k [\epsilon_\ell - a_\ell]\Big)\\
\end{align*}
}
The second ratio from Equation \ref{eqn:2nd_ratio_composition} is bounded by $a_1$, according to the Definition \ref{def:ab-influence}. Therefore, we have the ratio bounded by:
\begin{align*}
    \frac{\Pr[\bigwedge_{\ell=1}^k \mech_\ell(\Data)=\outp_\ell|\secret_i]}{\Pr[\bigwedge_{\ell=1}^k \mech_\ell(\Data)=\outp_\ell|\secret_j]}
    \leq & \exp\Big(\sum_{\ell=1}^k [\epsilon_\ell - a_\ell ]+a_1 \Big) \\
    \leq & \exp\Big(\sum_{\ell=1}^k \epsPuffer_\ell - \sum_{\ell=1}^k a_\ell + \max_l a_\ell \Big). 
\end{align*}
\end{proofE}

Theorem~\ref{the:markov-trick} and Theorem~\ref{the:sfc} together provide us with a powerful tool to translate differential privacy mechanisms into composable Pufferfish mechanisms.
With an $\ab$-influence curve, we can use any popular DP algorithms, e.g., randomized response~\cite{RR}, exponential mechanism~\cite{EM}, 
etc.,
depending on the use case, and not be restricted to the Laplace mechanism ~\cite{song2017composition, song2017pufferfish}, for example. Our theorems ensure that all of them would be valid Pufferfish mechanisms, and they would compose linearly.

Our results have several key implications.
First, this reduces the manual effort needed to design Pufferfish mechanisms for different use cases and applications.
Currently, there are not many existing (composable) Pufferfish mechanisms to choose from, unlike differential privacy, which has a rich set of well-understood mechanisms. Currently, using Pufferfish to create privacy-preserving data products is complex: the data curator needs to design new mechanisms, prove they satisfy Pufferfish individually, and then manually analyze their joint privacy leakage (compositional properties). This process, especially the last step, requires significant expenditure of manual effort.
However, our result implies that one can re-purpose existing differential privacy mechanisms to achieve Pufferfish. It requires a one-time (per application) computation of the \ab-curve, and then, any DP algorithm and mature DP frameworks like OpenDP \cite{opendp} can be used.
In fact, our theoretical results imply that any composable Pufferfish mechanism must meet DP-like constraints, so adapting differential privacy is essentially inevitable and not a heuristic decision. We show the practical benefit of this in our evaluation (Section~\ref{sec:exp}).
Finally, our result provides a retrospective understanding of what Pufferfish privacy is achieved for data that were already released under differential privacy.
\subsection{Practical Considerations} \label{sec:practical}
We next consider the following practical questions. Are there efficient ways to upper-bound the $\ab$-influence curve? How can we compute the curve when there is uncertainty about the data prior (e.g., the prior is mis-specified)? What happens if, instead of specifying the priors $\priorset$ directly, the data curator specifies them indirectly as ``the set of all priors having a given $\ab$-influence curve,'' but this curve is incorrect for the true distribution?

\subsubsection{Efficient Curve Estimation.}
\label{sec:curve_estimation}
Computing the $\ab$-influence curve \emph{exactly} may require a brute-force search over all index set splits $(\lowinf,\highinf)$ and all secret pairs $(\secret_i,\secret_j)$ in Definition~\ref{def:ab-influence} (for some well-defined priors, the compute can be much simpler, as we show in Section~\ref{sec:abcurve}). When this is infeasible, the goal is to efficiently compute an upper bound $\hat{a}(b)$ on the curve to get an upper bound on $\epsPuffer$.
A simple strategy is to sample a subset of the index set splits and only calculate the $\ab$-influence curve over that subset, which produces an upper bound because the true curve is an infimum over all the possible index splits.
While choosing a good subset may not always be straightforward, we present an efficient heuristic for graph-structured priors, where each secret corresponds to the value of a node in a Markov or Bayesian network (e.g., \cite{song2017composition}).
For a secret pair $(\secret_i,\secret_j)$ associated with a node $X_s$, and a small value of $b$, we construct candidate splits by placing $X_s$ and its $b-1$ nearest nodes into the high-influence region. Formally, let $\mathcal{N}_{b-1}(X_s)$ denote the collection of all possible sets of $b-1$ nearest nodes to $X_s$ in the graph. For each $N \in \mathcal{N}_{b-1}(X_s)$, we define $\highinf_b^N(X_s) = \{X_s\} \cup N$ and $\lowinf_b^N(X_s) = I \setminus \highinf_b^N(X_s)$. 
We use these candidate sets (with  $a^\dagger$ defined as in Definition \ref{def:ab-influence}) to compute $\hat{a}(b)$ as follows:
\begin{align*}
    \hat{a}^\star_{(\secret_i,\secret_j)}(b)
    &=
    \min_{N \in \mathcal{N}_{b-1}(X_s)}
    a^{\dagger}_{(\secret_i,\secret_j)}(\lowinf_b^N(X_s))\\
\text{ and }\hat{a}(b)
    &=
    \max_{(\secret_i,\secret_j)\in \secretset}
    \hat{a}^\star_{(\secret_i,\secret_j)}(b).
\end{align*}
This computation is manageable for small values of $b$.
The intuition behind this heuristic is that 
in graph-structured priors, nodes that are closest to $X_s$ are typically the most strongly correlated with it, and therefore their values are more likely to reveal information about the secret. 
We show an example of applying this efficient curve estimation heuristic to a Bayesian network in Section~\ref{sec:abcurve}.

\subsubsection{Uncertain or mis-specified prior.}
\label{sec:uncertainty}
Suppose a practitioner specifies a $\prior$ that differs from the true prior $\prior^*$, but is not too far off---for example, the probabilities have up to 20\% error (i.e., $0.8 \Pr[\Dataquilt = \dataquilt \mid \secret, \prior] \leq \Pr[\Dataquilt = \dataquilt \mid \secret, \prior^*] \leq 1.2\Pr[\Dataquilt = \dataquilt \mid \secret, \prior]$). We can still upper bound the $\ab$-influence curve in the presence of such uncertainties. In general, suppose there is some lower bound
$l(\dataquilt,\secret)$ (e.g., $0.8 \Pr[\Dataquilt = \dataquilt \mid \secret, \prior]$) and some upper bound $u(\dataquilt,\secret)$ (e.g., $1.2\Pr[\Dataquilt = \dataquilt \mid \secret, \prior]$) on $\Pr[\Dataquilt = \dataquilt \mid \secret, \prior^*]$
for all $\secret \in \secretset$ and $\Dataquilt$. 
Then, the value $a^\dagger_{(\secret_i,\secret_j)}$ from Definition~\ref{def:ab-influence} can be replaced by:
\begin{equation}
\begin{aligned}
    \overline{a}^\dagger_{(\secret_i,\secret_j)}(\mathcal{L})
= \max_{\dataquilt}
\max\!\Bigg\{
& \left|
\log
\frac{u(\dataquilt,\secret_i)}
     {l(\dataquilt,\secret_j)}
\right|,
\left|
\log
\frac{l(\dataquilt,\secret_i)}
     {u(\dataquilt,\secret_j)}
\right|
\Bigg\}.
\end{aligned}
\label{eq:uncertainty_ub}
\end{equation}
Using this upper bound $\overline{a}^\dagger_{(\secret_i,\secret_j)}$ in place of $a^\dagger_{(\secret_i,\secret_j)}$ provides an upper bound on the true $a(b)$-influence curve. 
As an example, we apply this technique to a Bayesian network in Section \ref{sec:abcurve}.

\subsubsection{Mis-specified $\ab$-influence Curve.} 
Suppose the practitioner specifies the prior set $\priorset$ as 
``the set of all priors whose $\ab$-influence curve is upper bounded by a given function $a_1(b)$'', but the actual influence curve for the true prior is some other function $a_2(b)$, which is not upper bounded by $a_1(b)$.
To get a target Pufferfish parameter $\epsPuffer$, the practitioner will incorrectly set $\epsDP=\frac{\epsPuffer - a_1(b^\star)}{b^\star}$ for some integer $b^\star$. What is the true Pufferfish parameter $\epsilon^\prime_{\text{Puffer}}$ that is achieved by this $\epsDP$-mechanism? Using Theorem \ref{the:markov-trick},
\begin{align*}
    \epsilon'_{\text{puffer}} &= \inf\limits_{b > 0} \big( a_2(b) + b \cdot \frac{\epsPuffer - a_1(b^\star)}{b^\star}\big) \\
    & \leq \big( a_2(b^\star) + b^\star \cdot \frac{\epsPuffer - a_1(b^\star)}{b^\star}\big) \\
    &= \epsPuffer + (a_2(b^\star) - a_1(b^\star)).
    \label{eq:uncertainty_ub}
\end{align*}
Thus, if the practitioner fails to specify a correct upper bound on the influence curve, 
the actual achieved privacy parameter might be weaker by at most the difference between the two curves at $b^\star$: $a_2(b^\star) - a_1(b^\star)$.

\subsection{Example $\ab$-influence Curves}\label{sec:abcurve}

Dependent DP~\cite{DependDP}, which assumes an item is correlated with at most L-1 other items, is a special case of the $\ab$-influence curve where $a(b)=\infty$ for $b < L$ and $a(b)=0$ for $b \ge L$. The connection between differential privacy and Pufferfish  with an independence assumption \cite{pufferfish} means that per-entry DP is equivalent to setting $a(0)=\infty$ and $a(b)=0$ for $b\geq 1$. Similarly, when there are $n$ attributes, using group DP to protect groups of $m^\prime$ records is equivalent to setting $a(b)=\infty$ for $b<m^\prime n$ and $a(b)=0$ otherwise. 
In other applications, priors are often specified as Markov and Bayesian networks \cite{song2017composition, song2017pufferfish,beinlich1989alarm}. Thus, we next provide examples of $\ab$-influence curve derivations for Markov chains, multivariate Gaussian priors, and real-world Bayesian networks with prior uncertainty.

\textbf{\emph{Markov Chain.}} 
The $\ab$-influence curve for a simple Markov chain can be calculated analytically as the following.
The proofs can be found in the \shortlong{full version \cite{nfcarxiv}}{appendix}.

\begin{example}\label{ex:markov}
We consider a binary and time-homogeneous chain with each node's value being 0 or 1. 
We consider the value of each node as a secret.
Its transition matrix is:
\[ P =
\begin{bmatrix}
p & 1-p\\
1-q & q
\end{bmatrix},
\qquad 0<p,q<1,\]
and the prior distribution of the Markov chain is its stationary distribution. We set the chain has length $T$, and $T$ is large enough so we only consider $b <<T$. 
\end{example}

In this case, we can have a closed-form solution for $\ab$-influence curve. 
Let us consider an arbitrary $i$-th node $X_i$'s value as a secret.
for a fixed $b$, we partition the nodes into $\highinf$ and $\lowinf$, such that $\highinf$ contains $X_i$ and $b-1$ surrounding nodes.
Let $d_L$ (resp., $d_R$) denote the distance from $X_i$ to the closest node of $\lowinf$ on the left (resp., right), so that $d_L+d_R-1=b$.
By the Markov property, conditioning on all variables in $\lowinf$ influences $X_i$ only through these two closest low-influence nodes, and the leakage through $\lowinf$ reduces to the leakage through those two nodes.

Thus, finding an optimal $\ab$-influence curve becomes a one-dimensional optimization problem of finding $(d_L,d_R)$ that maximizes the leakage.
The maximum leakage occurs when these nodes both have the same value of 0 or 1.
When $q>p$, the chain tends to stay at 1 more strongly than at 0, so the leakage is maximum when these nodes have values 0.
When $q\leq p$, the worst case is instead achieved when both have values $1$.
Finally, from midpoint convexity, the influence is minimized when the two nearest low-influence nodes are placed as evenly as possible on both sides of $X_i$, giving
$d_L^\star=\lfloor(b+1)/2\rfloor$ and $d_R^\star=\lceil(b+1)/2\rceil$.
In summary, the optimal $\ab$-influence curve can be calculated as:

\begin{theoremE}\label{thm:markov-example}
Given the binary Markov chain in Example~\ref{ex:markov}, 
let $\lambda=q+p-1$, the $\ab$-influence curve is calculated by 
\[
a^\star(b) = 
\Bigg|\log \frac{\pi_0+\lambda^{d_L}(1-\pi_0)}{\pi_0-\lambda^{d_L}\pi_0}\Bigg|+\Bigg|\log \frac{\pi_0+\lambda^{d_R}(1-\pi_0)}{\pi_0-\lambda^{d_R}\pi_0}\Bigg|\text{ if $q > p$},
\] \[
a^\star(b) = 
\Bigg|\log \frac{\pi_1+\lambda^{d_L}(1-\pi_1)}{\pi_1-\lambda^{d_L}\pi_1}\Bigg|+\Bigg|\log \frac{\pi_1+\lambda^{d_R}(1-\pi_1)}{\pi_1-\lambda^{d_R}\pi_1}\Bigg|\text{ if $q \leq p$},
\] 
where $d_L=\Big\lfloor\frac{b+1}{2}\Big\rfloor$ and $d_R=\Big\lceil\frac{b+1}{2}\Big\rceil$. 
The star ($^\star$) indicates that this curve is optimum (tight).
\end{theoremE}
\begin{proofE}
The binary Markov Chain has the stationary distribution as:
\[
\pi_0=\frac{1-q}{2-p-q},\qquad
\pi_1=\frac{1-p}{2-p-q}.
\]
Using the stationary distribution as the prior, the chain is reversible, and the reverse transition matrix is also $P$. 
Let $\lambda=p+q-1$,
the k-step transition matrix (forward or backwards) is represented as:
\[
P^k_{0,0}=\pi_0+\lambda^k(1-\pi_0), P^k_{1,0}=\pi_0-\lambda^k\pi_0, 
\]
\[
P^k_{1,1}=\pi_1+\lambda^k(1-\pi_1), P^k_{0,1}=\pi_1-\lambda^k\pi_1. 
\]
 
Assume the secret is the value of each node, i.e., the secret pair is $(X_i=0, X_i=1), \forall i$. 
Given $b$, we can find $\ab$ as follows.
For each $i$, naturally, the high-influence region $\highinf$ will be $b$ nodes around $X_i$, including itself. 
Let $d_L$, $d_R$ be the distance (in number of nodes) between $X_i$ and its nearest left and right nodes that are in the low-influence region $\lowinf$.
In other words, $\highinf$ is $X_i$, $d_L-1$ nodes to its left, and $d_R-1$ nodes to its right ($d_L+d_R-1=b$), and the rest are $\lowinf$.
By the Markov property, the information of $X_i$ leaking through the entire $\lowinf$ is equivalent to the information leaking through the two nodes $X_{i-d_L}$ and $X_{i+d_R}$ (left and right closest nodes in $\lowinf$). 
If we let $l, r$ represents the value of nodes $X_{i-d_L}, X_{i+d_R}$, respectively, from Definition~\ref{def:ab-influence}, 
$a^\star(b)$ is calculated by:
\[
a^\star(b) = \min_{\substack{d_L,d_R\ge 1\\ d_L+d_R-1=b}}\max_{r, l}
\bigl(\left|\log\frac{P^{d_L}_{0,l}}{P^{d_L}_{1,l}}\right|+\left|\log\frac{P^{d_R}_{0,r}}{P^{d_R}_{1,r}}\right|\bigr).
\]
We can see that $a^\star(b)$ does not depend on the absolute position $i$ (except for the boundary effects), and only depends on $d_L$ and $d_R$. Therefore, we only need to search for the pair $(d_L, d_R)$ that minimizes the above objective.

Now, we consider the condition which yields the maximum of  $\left|\log\frac{P^{d_L}_{0,l}}{P^{d_L}_{1,l}}\right|$; the right-side term (with $r$) can be analyzed in the same way due to symmetry.
Will the maximum occur when $l=0$ or $l=1$? 
If $l=0$ yields the maximum, then we must have:
\begin{align*}
     \Bigg|\log\frac{P^{d_L}_{0,0}}{P^{d_L}_{1,0}}\Bigg| &> \Bigg|\log\frac{P^{d_L}_{0,1}}{P^{d_L}_{1,1}}\Bigg| \\
 \Bigg|\log \frac{\pi_0+\lambda^{d_L}(1-\pi_0)}{\pi_0-\lambda^{d_L}\pi_0}\Bigg| &> \Bigg|\log \frac{\pi_1+\lambda^{d_L}(1-\pi_1)}{\pi_1-\lambda^{d_L}\pi_1}\Bigg|,
\end{align*}
which simplifies to the condition $q > p$. Therefore, when $q > p$, the choice $l=0$ maximizes the left log-ratio term; when $q\leq p$, the choice $l=1$ maximizes it. 
The choice for $r$ is the same as $l$ depending on the situation.
Next, find the optimal $d_L, d_R$. We observe that the log ratio  $\Bigg |\log \frac{\pi_0+\lambda^k(1-\pi_0)}{\pi_0-\lambda^k\pi_0}\Bigg|$ and $\Bigg|\log \frac{\pi_1+\lambda^k(1-\pi_1)}{\pi_1-\lambda^k\pi_1}\Bigg|$ are convex functions with respect to $k$. From midpoint convexity, the optimal $d_L, d_R$ must be,
\[
d_L=\lfloor (b+1)/2\rfloor, d_R=\lceil (b+1)/2\rceil. 
\]
\end{proofE}

\begin{figure}[t]
    \centering
    \includegraphics[width=0.32\linewidth]{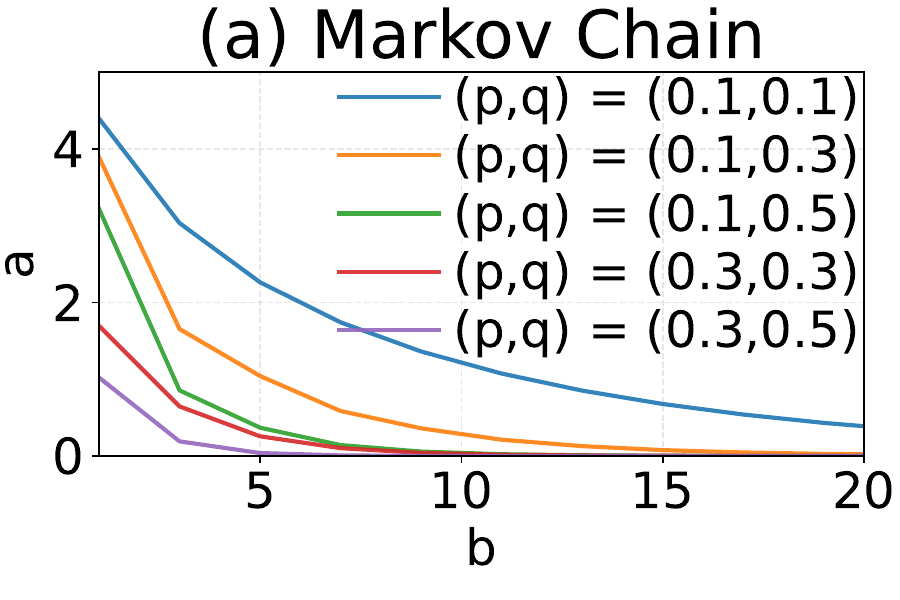}
    \includegraphics[width=0.32\linewidth]{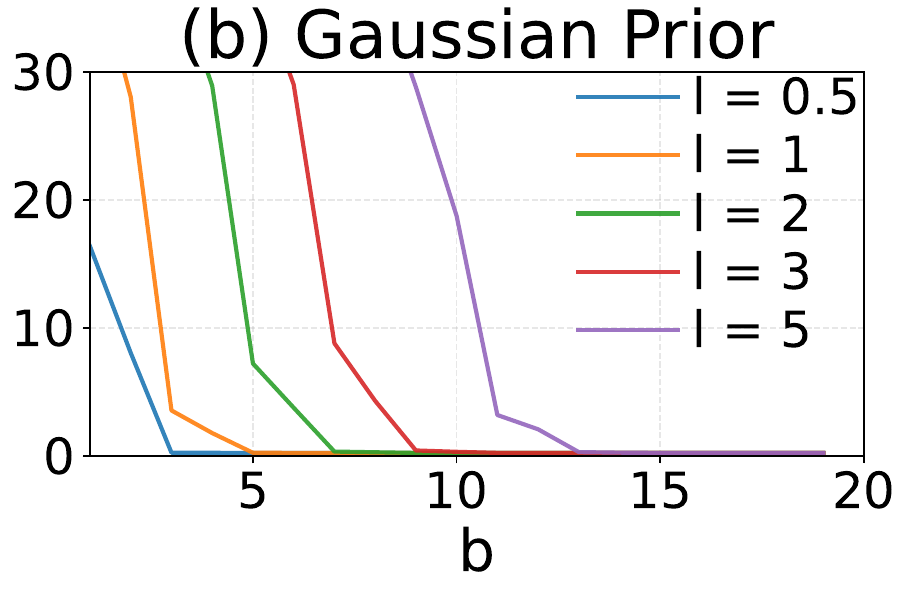}
    \includegraphics[width=0.32\linewidth]{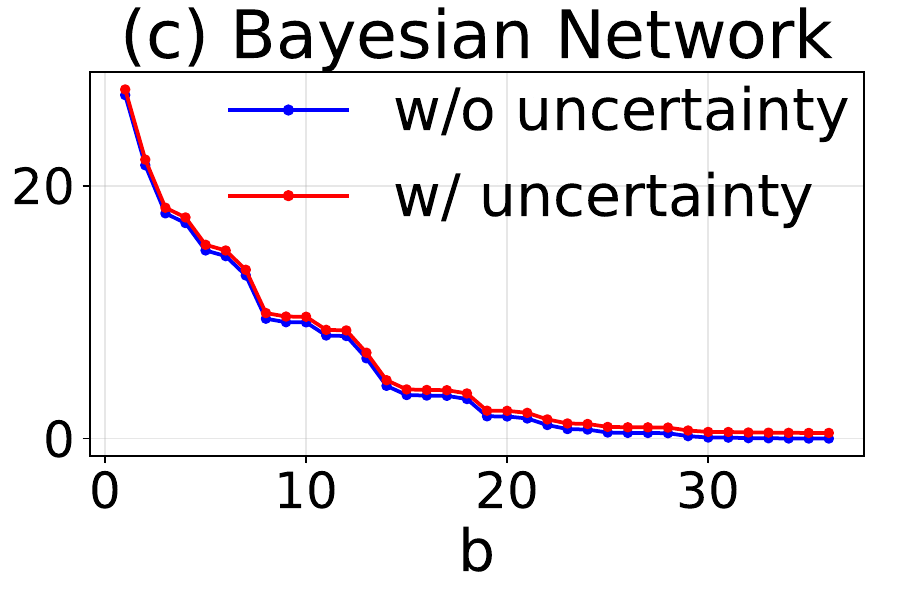}
    \caption{$\ab$-influence curves under different priors: (a) binary Markov chains and (b) multivariate Gaussian priors. (c): $a(b)$ influence curve for ALARM network. The blue curve is calculated from the exact prior, and the red curve is an upper bound by assuming 20\% prior uncertainty.
    }
    \label{fig:ab-curve}
\end{figure}


Figure~\ref{fig:ab-curve}(a) plots the resulting $\ab$-influence curves for different values of $p$ and $q$.
When $p$ and $q$ are close to $1/2$ (i.e., the chain exhibits weak temporal correlation), the resulting $\ab$-curve lies closer to the origin and quickly decays---indicating that the secret does not propagate across nodes much.
As $p$ or $q$ decreases towards 0 (or increases towards 1, which we omit due to symmetry), the chain exhibits stronger temporal correlation, and the $\ab$-influence curve also spreads out over $b$.

\textbf{\emph{Multivariate Gaussian.}}
Next, we illustrate how we can derive the $\ab$-influence curve when the prior distribution $\prior$ is a multivariate Gaussian.
While a closed-form solution is not possible, we can still efficiently calculate the optimal $\ab$-influence.

\begin{example}\label{ex:gaussian}
Consider a single-record continuous dataset represented by a vector $\mathbf{x}\in\mathbb{R}^n$.
Throughout this example, the prior distribution is a truncated multivariate Gaussian,
\[
\mathbf{x}\sim \mathcal{N}(0,\Sigma)\  \ \mathbf{x}\in[-\gamma,\gamma]^n,
\]
where $\gamma>0$ bounds the support.
The covariance matrix $\Sigma$ is:
\[
\Sigma_{jk} \;=\; \exp\!\left(-\frac{(j-k)^2}{\ell}\right),\qquad j,k\in[n].
\]
The secret is defined as \emph{region events} on each coordinate.
Specifically, the secret is whether each coordinate $x_i$'s value is in a certain range:
\[
\secret_i(r)\ =\ \{x_i\in[r,r+\delta],\qquad r\in[-\gamma,\gamma-\delta]\}.
\]
A secret pair on $i$ is given by two disjoint region events $\bigl(\secret_i(r),\,\secret_i(r')\bigr)$, where
$[r,r+\delta]\cap[r',r'+\delta]=\emptyset$.

\end{example}

Among all the entries, the middle coordinate $x_{i^\star}$ has the greatest information leakage, where $i^\star=\lceil n/2\rceil$. Therefore, we only need to consider secrets defined on $x_{i^\star}$, $s(r)\equiv\{x_{i^{\star}}\in[r,r+\delta]\}$. For a fixed $b$, the optimal low-influence index set $\lowinf^{\star}$ with size $n-b$ is the set of indices with the farthest distance to $i^\star$. Under $\lowinf^\star$, we calculate $a^\star(b)=\max_{s(r), s(r'), x_\lowinf} \log\frac{\Pr[x_\lowinf|x_{i^{\star}}\in [r,r+\delta]]}{\Pr[x_\lowinf|x_{i^{\star}}\in [r',r'+\delta]]}$, which is the optimal (tight) $\ab$. Although $\mathbf{x}_{\lowinf}$ is high-dimensional, we have:
\[
\Pr(\mathbf{x}_{\lowinf}\mid x_{i^\star})
\propto
\Pr(x_{i^\star}\mid \mathbf{x}_{\lowinf})/\Pr(x_{i^\star}).
\]
The conditional distribution of $\Pr(x_{i^\star}\mid \mathbf{x}_{\lowinf})$ is
a one-dimensional Gaussian distribution, denoted as $\mathcal{N}(\mu_{i^\star,\lowinf},\,\sigma^2_{i^\star,\lowinf})$.
As $\sigma^2_{i^\star,\lowinf}$ does not depend on $\mathbf{x}_{\lowinf}$, for fixed $\lowinf$, the maximization over $\mathbf{x}_{\lowinf}$ reduces to sweeping over $\mu_{i^\star,\lowinf}$.

\begin{theorem}\label{thm:gaussian-gp-example}
Consider Example~\ref{ex:gaussian} where $\mathbf{x}\sim\mathcal{N}(0,\Sigma)$ is truncated to $[-\gamma,\gamma]^n$ and $\Sigma$ is induced by the Gaussian-process kernel $\Sigma_{jk}=\exp\!\bigl(-(j-k)^2/\ell\bigr)$.
Fix a region length $\delta<<\gamma$ and define region secrets as $\secret_i(r)\equiv\{x_i\in[r,r+\delta]\}, r\in[-\gamma, \gamma-\delta]$ on any coordinate $i$. The effective set of secret pairs to search for is defined on the center coordinate $x_{i^\star}$ with $i^\star=\lceil n/2\rceil$.
Moreover, for each fixed $b$:
(1) the optimal low-influence index set $\lowinf$ with size $n-b$ is the set of indices with the farthest distance to $i^\star$, and 
(2) the max influence of $\mathbf{x}_{\lowinf}$ over $x_{i^\star}$ reduces to a one-dimensional sweep over the mean on the conditional distribution $x_{i^\star}|\mathbf{x}_\lowinf$. 
\end{theorem}
\begin{proofE}
    Given any entry $x_i$ and fixing $b$, 
the information leakage of $x_i$ through $\lowinf$ is quantified by the maximum influence of $\mathbf{x}_{\lowinf}$ on $x_i$, which works on the probability distribution 
$\Pr(\mathbf{x}_{\lowinf} \mid x_i \in [r_i,r_i+\delta])$, a high-dimensional Gaussian distribution. 
To evaluate this quantity, we equivalently work on
$\Pr(x_i \in [r_i,r_i+\delta] \mid \mathbf{x}_{\lowinf})\Pr( \mathbf{x}_{\lowinf})/\Pr(x_{i}\in [r_i,r_i+\delta])$. The conditional distribution of $x_i$ given $\mathbf{x}_{\lowinf}$ is one-dimensional Gaussian,
$Normal(\mu_{i,\lowinf}, \sigma_{i,\lowinf}^2)$,
where $\mu_{i,\lowinf}$ denotes the conditional mean of $x_i$ induced by the observed values $\mathbf{x}_{\lowinf}$, and $\sigma_{i,\lowinf}^2$ is the corresponding conditional variance.
We have, 
\[
\mu_{i, \lowinf} = \Sigma_{i,\lowinf}\Sigma_{\lowinf,\lowinf}^{-1}\mathbf{x}_{\lowinf}, \quad
\sigma_{i, \lowinf}^2 = \Sigma_{ii} - \Sigma_{i,\lowinf}\Sigma_{\lowinf,\lowinf}^{-1}\Sigma_{\lowinf,i}.
\]
We can see that $\mu_{i, \lowinf}$ is a variable that changes with $x_\lowinf$, but $\sigma_{i, \lowinf}$ is a constant. Therefore, instead of sweeping all possible $x_\lowinf$, we can instead search for one-dimension $\mu_{i, \lowinf}$ to obtain $a^\star_i(b)$:

\begin{align*}
a^\star_i(b) &= \max_{r_i, r_i'} \inf_{\lowinf}\sup_{\widetilde{x}_\lowinf}\log \frac{\Pr(\mathbf{x}_\lowinf = \widetilde{x}_\lowinf | x_i \in [r_i,r_i+\delta])}{\Pr(\mathbf{x}_\lowinf = \widetilde{x}_\lowinf | x_i \in [r_i',r_i'+\delta])}\\
&= \max_{r_i, r_i'}  \inf_{\lowinf} \sup_{\widetilde{x}_{i, \lowinf}}\log \frac{\Pr(x_i \in [r_i,r_i+\delta]\mid \mathbf{x}_\lowinf = \widetilde{x}_\lowinf)}{\Pr(x_i \in [r_i',r_i'+\delta]\mid \mathbf{x}_\lowinf = \widetilde{x}_\lowinf)} \\
&\phantom{= \max_{r_i, r_i'}  \inf_{\lowinf} \sup_{\widetilde{x}_{i, \lowinf}}\log} - \log \frac{\Pr(x_i \in [r_i,r_i+\delta])}{\Pr(x_i \in [r_i',r_i'+\delta])} \\
&= \max_{r_i, r_i'}  \inf_{\lowinf} \sup_{\mu_{i, \lowinf}}\log \frac{\Pr(x_i \in [r_i,r_i+\delta]\mid \mu_{i, \lowinf})}{\Pr(x_i \in [r_i',r_i'+\delta]\mid \mu_{i, \lowinf})} \\
&\phantom{= \max_{r_i, r_i'}  \inf_{\lowinf} \sup_{\widetilde{x}_{i, \lowinf}}\log} - \log \frac{\Pr(x_i \in [r_i,r_i+\delta])}{\Pr(x_i \in [r_i',r_i'+\delta])}.
\end{align*}
The probability of the secret event $x_i \in [\alpha,\beta]$ can be expressed as
\begin{align*}
   & \Pr(x_i \in [\alpha,\beta]\mid \mu_{i, \lowinf}) 
= \Phi\!\left(\frac{\beta-\mu_{i, \lowinf}}{\sigma_{i, \lowinf}}\right)
  - \Phi\!\left(\frac{\alpha-\mu_{i, \lowinf}}{\sigma_{i, \lowinf}}\right),\\
 & \Pr(x_i \in [\alpha,\beta]) 
= \Phi\!\left(\frac{\beta-\mu_i}{\sigma_{i}}\right)
  - \Phi\!\left(\frac{\alpha-\mu_i}{\sigma_{i}}\right),
\end{align*}
where $\Phi$ denotes the standard normal CDF. 

Under the Gaussian Process Kernel in example~\ref{ex:gaussian}, we can derive a much more efficient approach to calculate the $\ab$-influence curve. In this case, for each $x_i$, the index set $\lowinf$ is optimized at the least correlated entries, i.e., the $b$ smallest entries in magnitude of the $i$-th row of the covariance matrix.
\end{proofE}

We report $\ab$-curve for decay factor $\ell=\{0.5, 1,2,3,5\}$, and pick $\gamma=5, \delta=0.1$. 
The results are shown in Figure~\ref{fig:ab-curve}(b). As $\ell$ grows larger, the correlation among entries becomes stronger, and the corresponding $\ab$-curve is farther away from the origin. 

\begin{table*}[h]
\small
\setlength{\tabcolsep}{3.5pt} 
\begin{tabular}{ccccccccccccc}
 & \multicolumn{3}{c}{\textbf{\ourmech} w/o Uncertainty} & \multicolumn{3}{c}{MQM w/o Uncertainty} & \multicolumn{3}{c}{\textbf{\ourmech} w/ Uncertainty} & \multicolumn{3}{c}{MQM w/ Uncertainty} \\ \cmidrule(r){2-4} \cmidrule(r){5-7} \cmidrule(r){8-10} \cmidrule(r){11-13}
 $\epsPuffer$   & Acc@1 & Acc@2 & Acc@3 & Acc@1 & Acc@2 & Acc@3 & Acc@1 & Acc@2 & Acc@3 & Acc@1 & Acc@2 & Acc@3 \\ \hline
0.5 & \textbf{79.0\%} & \textbf{63.1\%} & \textbf{58.5\%} & 36.3\% & 26.7\% & 20.83\% & \textbf{42.4\%} & \textbf{34.0\%} & \textbf{28.4\%} & 16.2\% & 11.1\% & 8.1\% \\
1 & \textbf{88.5\%} & \textbf{76.7\%} & \textbf{72.5\%} & 46.4\% & 36.3\% & 30.6\% & \textbf{81.8\%} & \textbf{66.9\%} & \textbf{62.4\%} & 38.6\% & 28.9\% & 23.0\% \\ 
2 & \textbf{94.0\%} & \textbf{85.7\%} & \textbf{82.8\%} & 60.8\% & 47.5\% & 43.0\% & \textbf{92.8\%} & \textbf{82.9\%} & \textbf{79.6\%} & 55.2\% & 43.6\% & 38.8\% \\ 
3 & \textbf{95.7\%} & \textbf{89.6\%} & \textbf{87.7\%} & 72.1\% & 55.3\% & 50.8\% & \textbf{95.1\%} & \textbf{88.3\%} & \textbf{86.1\%} & 68.1\% & 52.3\% & 47.9\% \\ 
4 & \textbf{96.7\%} & \textbf{91.7\%} & \textbf{90.4\%} & 79.1\% & 61.7\% & 57.0\% & \textbf{96.4\%} & \textbf{91.0\%} & \textbf{89.5\%} & 76.7\% & 59.2\% & 54.6\% \\ 
5 & \textbf{97.3\%} & \textbf{93.0\%} & \textbf{92.1\%} & 83.4\% & 66.9\% & 62.0\% & \textbf{97.1\%} & \textbf{92.6\%} & \textbf{91.5\%} & 81.9\% & 65.0\% & 60.1\% \\
\end{tabular}
\caption{Acc@k for evaluating the Top-3 popular visited locations using the Foursquare check-in dataset, under various $\epsPuffer$. The metric is higher the better, and the best mechanism for each row is highlighted in bold.
\rebuttal{Our new mechanism (\ourmech) outperforms both MQM and group DP-based approaches in all privacy parameters.}{
We also report the numbers when 20\% uncertainty is set around the prior (``w/ Uncertainty'').}
Our new mechanism (\ourmech) outperforms MQM in all privacy parameters.
}
\label{tab:pos-acc}
\end{table*}

\begin{table*}[h]
    \centering
    \small
    \setlength{\tabcolsep}{4pt} 
    \begin{tabular}{ccccccccccccc}
 & \multicolumn{3}{c}{\textbf{\ourmech} w/o Uncertainty} & \multicolumn{3}{c}{MQM w/o Uncertainty} & \multicolumn{3}{c}{\textbf{\ourmech} w/ Uncertainty} & \multicolumn{3}{c}{MQM w/ Uncertainty} \\
 \cmidrule(r){2-4} \cmidrule(r){5-7} \cmidrule(r){8-10} \cmidrule(r){11-13}
 $\epsPuffer$ & HR (↑) & NDCG (↑) & $\ell_1$ (↓) & HR (↑) & NDCG (↑) & $\ell_1$ (↓) & HR (↑) & NDCG (↑) & $\ell_1$ (↓) & HR (↑) & NDCG (↑) & $\ell_1$ (↓) \\  
 \midrule
0.5 & \textbf{78.7\%} & \textbf{0.937} & \textbf{765} & 37.6\% & 0.726 & 5500 & \textbf{45.6\%} & \textbf{0.763} & \textbf{3838} & 18.6\% & 0.606 & 17842 \\
1   & \textbf{89.3\%} & \textbf{0.977} & \textbf{295} & 49.8\% & 0.791 & 3025 & \textbf{81.7\%} & \textbf{0.952} & \textbf{634} & 40.4\% & 0.742 & 4746 \\
2   & \textbf{94.4\%} & \textbf{0.987} & \textbf{105} & 64.3\% & 0.868 & 1554 & \textbf{93.1\%} & \textbf{0.986} & \textbf{147} & 59.3\% & 0.841 & 1976 \\
3   & \textbf{96.3\%} & \textbf{0.993} & \textbf{69}  & 73.3\% & 0.917 & 913  & \textbf{95.5\%} & \textbf{0.991} & \textbf{83}  & 70.1\% & 0.900 & 1125 \\
4   & \textbf{97.2\%} & \textbf{0.994} & \textbf{60}  & 79.2\% & 0.946 & 571  & \textbf{96.9\%} & \textbf{0.994} & \textbf{61}  & 77.1\% & 0.933 & 684 \\
5   & \textbf{97.8\%} & \textbf{0.994} & \textbf{56}  & 83.4\% & 0.964 & 386  & \textbf{97.6\%} & \textbf{0.994} & \textbf{57} & 81.9\% & 0.957 & 449 \\
    \end{tabular}
    \caption{
    Hit Rate@K (HR), NDCG@K (NDCG), and $\ell_1$ Count Error ($\ell_1$) for evaluating the Top-3 popular visited locations using the Foursquare check-in dataset, under various $\epsPuffer$.
    Hit Rate@K and NDCG@K are higher the better (indicated by the ↑ symbol), and $\ell_1$ Count Error is lower the better (indicated by the ↓ symbol). The best among the mechanisms are highlighted in bold.
    \rebuttal{Our new mechanism (\ourmech) outperforms both MQM and group DP-based approaches in all privacy parameters and all metrics.}{
We also report the numbers when 20\% uncertainty is set around the prior (``w/ Uncertainty'').}
    Our new mechanism (\ourmech) outperforms MQM in all privacy parameters and all metrics.
    }
    \label{tab:hit-rate-ndcg}
\end{table*}

\textbf{\emph{Bayesian Network with Uncertain Prior.}}
Finally, we study a realistic use-case having a general Bayesian network as the prior, with uncertainty added as in Section \ref{sec:uncertainty}.
ALARM~\cite{beinlich1989alarm} is a benchmark for patient monitoring.
It uses an expert-constructed Bayesian network, which models the medical knowledge of dependencies between diagnostic hypotheses, intermediate physiological states, and observable monitoring variables through 37 nodes, 46 directed edges, 16 finding variables, and 13 intermediate variables.
All nodes are categorical variables representing 
clinical conditions, physiological states, or measurements. 
In our setup, the secret is the value of an individual node in a patient record.

Figure~\ref{fig:ab-curve}(c) plots the calculated $\ab$-influence curve (blue line), using the upper-bound estimation method from Section~\ref{sec:curve_estimation} (there are too many nodes to calculate the exact curve). 
Additionally, we calculated a curve for the case of prior uncertainty using the techniques of Section \ref{sec:uncertainty} (red line).
Specifically, we used 20\% prior uncertainty, i.e., $u(\dataquilt, \secret) = 1.2p(\dataquilt, \secret), l(\dataquilt, \secret) = 0.8p(\dataquilt, \secret)$.
Even with this moderate amount of fluctuation, the $\ab$-influence curve does not change much to account for potential uncertainty.

\section{Experiments}\label{sec:exp}

Our experiments highlight the benefit of our approach: (1) they rely on composition, (2) they take into account uncertainty about data-generating priors, and (3) they demonstrate that the freedom of adapting existing DP algorithms often provides a quick way of improving over hand-crafted Pufferfish mechanisms.

\subsection{Experimental Setup}

\subsubsection{Datasets}
We evaluate the mechanisms with two representative datasets where strong correlations exist between datapoints.

\textbf{Foursquare check-in dataset}~\cite{yang2016participatory, yang2015nationtelescope} contains users' location data with a corresponding timestamp across 77 countries.
Each location is associated with a semantic label, such as ``Bar'', ``Hotel'', or ``Restaurant''. We merged similar labels (e.g., ``Cocktail Bar'' with ``Bar'') and only chose the 77 most common labels worldwide, mapping all other locations as ``Other''.
Each datapoint of the resulting dataset is location information (one of the 77 labels or ``Other'') of a particular user at a particular time.
The rules for merging labels are in \shortlong{the Appendix of~\cite{nfcarxiv}}{Appendix~\ref{appendix:labelmap}}. 
Clearly, locations of a user within a short time period are expected to be highly correlated.

\textbf{Capture24 dataset}~\cite{chan2021capture,chan2024capture} contains 151 participants' activity data monitored across 24 hours.
Participants are grouped by age (18–29, 30–37, 38–52, 53+). The dataset annotated activities in multiple ways, and we use the \texttt{WillettsMET2018} annotation, where participants' activities are labelled into 11 activity categories.
Again, activities of a participant within a nearby time period are expected to be highly correlated.

\subsubsection{Query of Interest}

As a demonstration, we study \emph{noisy Top-$K$} queries.
In particular, for the Foursquare check-in dataset, the query we evaluate is: \emph{``Within each country, what are the Top-3 most frequently visited locations?''}.
For the Capture24 dataset. the query we evaluate is: \emph{``Within each age group, what are the Top-3 most frequent activities?''}.
These are deliberately chosen to evaluate the types of queries that prior work~\cite{song2017pufferfish} did not consider. For the queries that prior works studied (e.g., querying the average value~\cite{song2017pufferfish}), mechanisms from prior works can be used. 

\begin{figure}[h]
    \centering
    \includegraphics[width=\linewidth]{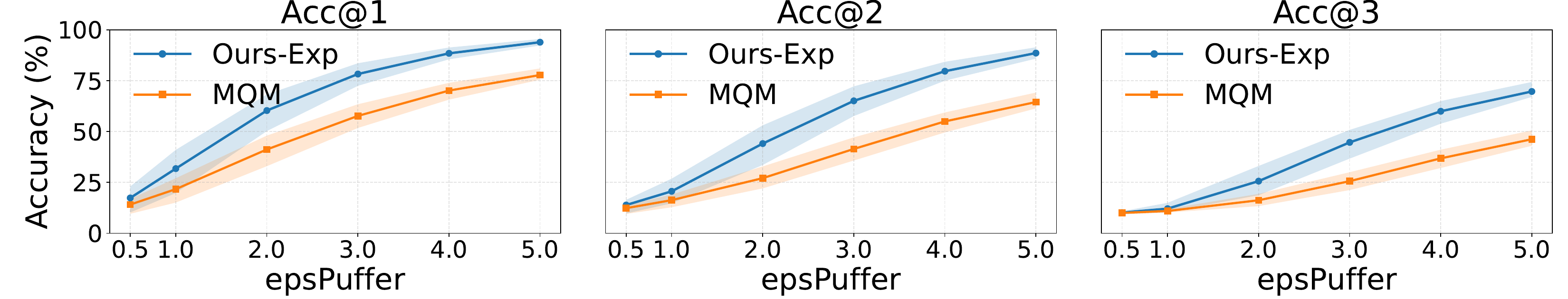}

    \caption{Acc@k for evaluating the Top-3 most frequent activities for each participant in the Capture24 dataset under various $\epsPuffer$. \ourmech outperforms MQM across all privacy parameters. Shaded regions indicate the utility range when assuming 20\% uncertainty around the prior.}
    \label{fig:capture}
\end{figure}

\subsubsection{Evaluated Mechanisms}

We compare an approach (\textbf{\ourmech}) based on our framework to MQM \cite{song2017pufferfish}.

\textbf{\ourmech} uses our framework's flexibility in choosing appropriate mechanisms by adapting the exponential mechanism \cite{EM} of differential privacy and running it $k$ times. The item returned in the first run is treated as the top item. The second run samples among the remaining items to produce the estimated second-ranked item, etc. Given a target Pufferfish privacy parameter $\epsPuffer$, the privacy budgets are allocated evenly among the $k$ runs so that the overall parameter via composition is $\epsPuffer$. 

\textbf{MQM} is a Pufferfish mechanism proposed by prior work~\cite{song2017pufferfish}. It adapts the Laplace mechanism of differential privacy to achieve Pufferfish privacy with composition.
Noise is added to the $m$ categories and the top-$k$ noisy categories are selected.

\subsubsection{Prior Distribution Construction} Using MQM or our method (\ourmech) requires specifying a  prior $\prior$.
Following prior works~\cite{song2017pufferfish}, we assume the temporal data are generated by a first-order Markov chain, and use the data to fit a Markov model with a stationary transition matrix.
For the Foursquare check-in dataset, each user's location is a state, and a location change is considered a state transition.
We truncate the dataset to only leave the check-in data for the first 500 days.
We built a $78\times 78$ transition matrix using the first 400 out of 500 days of data, and used the remaining 100 days for testing.
For the Capture24 dataset, each participant's activity is a state, and activity change is considered a state transition.
We hold out 8 participants from each age group for testing and use the rest to fit the $11\times11$ transition matrix.
Note that MQM requires the prior $\prior$ to be a Markov chain, while our general approach from Section~\ref{sec:sufficient} works with an arbitrary prior.
We additionally ran experiments with 20\% uncertainty added around the prior using the results of Section \ref{sec:uncertainty} (i.e., $u(\dataquilt, \secret) = 1.2p(\dataquilt, \secret), l(\dataquilt, \secret) = 0.8p(\dataquilt, \secret)$.). We denote results with uncertainty as ``w/ Uncertainty'' and results with no uncertainty as ``w/o Uncertainty''. 

\subsubsection{Utility Metrics}
We evaluate with the following metrics:

\textbf{Acc@k} measures the accuracy of predicting the $k$-th popular category correctly. The exact ranking (i.e., $k$) must be predicted correctly (e.g., if the most popular location is predicted as second-most popular, it is regarded as incorrect).

\textbf{Hit Rate@K} measures the precision of the correct Top-$K$ categories appearing in the predicted Top-$K$. The order does not matter, as long as the correct category appears in the predicted Top-$K$.

\textbf{NDCG@K}, or Normalized Discounted Cumulative Gain, is a frequently used metric in ranking systems~\cite{jarvelin2002cumulated}.
NDCG accounts for the order of the answer and rewards putting higher-ranked categories near the top, even if the exact rank was incorrect.
Precisely, NDCG is calculated by:
    \[
   \text{DCG@}K=\sum_{k=1}^{K}\frac{\text{rel}(\hat{\pi}_k)}{\log_2(k+1)},\quad
    \text{NDCG@}K=\frac{\text{DCG@}K}{\text{IDCG@}K},
    \]
where $\text{rel}(i)$ is the ground-truth score for item $i$ (we use true counts), and IDCG is the DCG of the ideal ranking.
By normalizing to the best possible ranking for the same query (IDCG), scores lie in $[0,1]$, and a higher score means better prediction.

\textbf{$\ell_1$ Count Error} is the sum of the absolute difference between the predicted Top-$K$'s counts and the true Top-$K$'s counts. 
We additionally included this metric because it can give us a sense of how different the actual counts are when misprediction occurs. For example, if the most and second-most popular categories have similar counts, this metric does not penalize swapping their ranking in prediction; however, if their counts differ significantly, getting the order wrong would significantly increase the error.
Specifically, the error is calculated by $\text{L1} \;=\; \sum_{k=1}^K \big| \tilde{c}_k - c_k \big|$, 
where $\tilde{c}_k$ is the count of the category that the noisy mechanism places at rank $k$, and $c_k$ is the count of the category that truly belongs at rank $k$. 

\textbf{Majority} We additionally show another example of what the accuracy would be if we use the Top-1/2/3 labels from the train data as the prediction to evaluate the complexity of the question. The Majority only achieved 18.61\% and 16.67\%  accuracy averaged across 3 positions. 
A full overview of the datasets, prior constructions, and implementation choices is provided in \shortlong{Appendix in ~\cite{nfcarxiv}}{Appendix~\ref{appendix:exp}}.

\subsection{Results: Foursquare Check-in Dataset}

Table~\ref{tab:pos-acc} (Acc@k) and Table~\ref{tab:hit-rate-ndcg} (Hit Rate@K, NDCG@K, $\ell_1$ Count Error) show the utility evaluation results for MQM and our new mechanism (\ourmech).
Both Table~\ref{tab:pos-acc} and Table~\ref{tab:hit-rate-ndcg} show that, across all privacy parameters and metrics we studied, \ourmech outperformed MQM both with and without prior uncertainty. 
\ourmech outperforms the Laplace-based MQM method because it uses the exponential mechanism of differential privacy, which is known to perform better for Top-$k$ queries~\cite{dwork2006calibrating}, while using our Theorem~\ref{the:markov-trick} to translate it into a composable Pufferfish mechanism.
Also, incorporating 20\% uncertainty into the computation of the worst-case leakage $\ab$-influence curve only mildly affects utility in most cases. Except for $\epsPuffer=0.5$, the utility decrease is modest, suggesting that our framework can preserve useful accuracy even when the prior is only approximately specified.

It is worth stating explicitly that the utility improvement in our experiments comes from our framework enabling the use of alternative DP mechanisms that are better matched to the query (the exponential mechanism for Top-$K$ queries).

\subsection{Results: Capture24 Dataset}

Figure ~\ref{fig:capture} shows the utility of \ourmech for various privacy parameters and metrics for the Capture24 dataset.
Again, in all the tested setups, \ourmech outperformed or showed on-par utility compared to the baselines (MQM; \ourmech also outperformed group DP, which is omitted). Note that this dataset only has 11 activities, so 9--10\% accuracy is almost random guessing.

\section{Conclusions and Future Work}\label{sec:conc}
In this paper, we study how to add composition properties to the Pufferfish framework.
We analyzed privacy collapses under composition and provided necessary
 and sufficient conditions for linear composition. We conclude that the constraints on mechanisms
required by the $\epsPuffer$-Pufferfish framework should be augmented with $\epsDP$-differential privacy 
constraints. We introduced the $\ab$-influence curve to translate between those two
privacy parameters.
Future work includes algorithm design under various classes of priors and extending the
necessary and sufficient conditions for privacy definitions like Renyi-Pufferfish.

\begin{acks}
This work was supported by the US National Science Foundation under Awards CNS-2349610 and CNS-2317232.
Any opinions, findings, and conclusions or recommendations expressed in this material are those of the author(s) and do not necessarily reflect the views of the National Science Foundation.
\end{acks}


\bibliographystyle{ACM-Reference-Format}
\bibliography{ref}

\shortlong{}{\appendix
\onecolumn
\section{Existing Posterior-Based Privacy Definition}

Table ~\ref{tab:privacy-lit-detail} summarizes existing inferential/posterior-based privacy definitions. We adopt the Pufferfish Privacy~\cite{pufferfish} framework in our analysis. 

\begin{table*}[]
\caption{A detailed overview of existing privacy definitions. We show that Pufferfish privacy can generalize to these existing privacy definitions. }
\label{tab:privacy-lit-detail}
\small
\renewcommand{\arraystretch}{1.5}
\begin{tabular}{m{0.07\linewidth}|m{0.15\linewidth}|m{0.15\linewidth}|m{0.2\linewidth}|m{0.25\linewidth}|>{\centering\arraybackslash}m{0.08\linewidth}}
\toprule
                            & Privacy Definition                                & Secret                        & Secret Pair                                                          & Prior Knowledge                 & Composition  \\ \midrule
\multirow{4}{*}{Likelihood} & Differential Privacy~\cite{dwork2006calibrating}                    & $\secret_i: \data_{i}=\alpha$ &$(\secret_i, \secret_j): \data_{i}=\alpha \text{ vs } \data_{i}=\beta$&  Each data                                record is independent& Yes \\ 
                            & Dependent Differential Privacy~\cite{DependDP}    & $\secret_i: \data_{i}=\alpha$ &$(\secret_i, \secret_j): \data_{i}=\alpha \text{ vs } \data_{i}=\beta$&  The set of distributions satisfying a specific probabilistic dependence relationship $\mathcal{R}$. & Yes\\ 
                            & Noiseless Privacy~\cite{noiseless}                & $\secret_i: \data_{i}=\alpha$ &$(\secret_i, \secret_j): \data_{i}=\alpha \text{ vs } \data_{i}=\beta$&  A specific distribution (e.g., i.i.d. Uniform or Gaussian) assumed to generate the data & Yes\\
                            & Blowfish Privacy~\cite{blowfish}                  & Any arbitrary propositional statement about tuple $i$. & Pairs of mutually exclusive secrets about the same individual & Distributions conditioned on a set of publicly known deterministic constraints $\mathcal{Q}$ & Yes \\ \midrule
\multirow{6}{*}{Posterior}  &Pufferfish Privacy~\cite{pufferfish}               & Customizable set of potential secrets $\secret$ & Customizable set of discriminative pairs $\secret_{pairs} \subseteq \secret \times \secret$ & Any customizable set of potential prior distributions $\theta$ & No \\
                            & Inferential Privacy~\cite{IP}                     & $\secret_i: \data_{i}=\alpha$ &$(\secret_i, \secret_j): \data_{i}=\alpha \text{ vs } \data_{i}=\beta$& Any customizable set of potential prior distributions $\theta$ & No \\
                            & Conditional Inference Privacy~\cite{CIP}          & A user's exact location coordinates at specific timestamps. & Any two potential locations that are within a specific geographic radius $r$. & Conditional Prior Class of distributions that model temporal dependence (smooth trajectories) & No\\
                            & Partial Knowledge Differential Privacy~\cite{PKDP}& $\secret_i: \data_{i}=\alpha$ &$(\secret_i, \secret_j): \data_{i}=\alpha \text{ vs } \data_{i}=\beta$&  A set of distributions $\Theta$ that can be factorized into independent parameters $\phi_0$ (sensitive value) and $\phi_{rest}$ (partial knowledge) & No\\
                            & Bayesian Differential Privacy~\cite{BDP}          & $\secret_i: \data_{i}=\alpha$ &$(\secret_i, \secret_j): \data_{i}=\alpha \text{ vs } \data_{i}=\beta$& A specific family of distributions where records are statistically dependent & No    \\
                            & Distribution Differential Privacy~\cite{DistDP}   & $\secret_i: \data_{i}=\alpha$ & $(\secret_i, \secret_j): \data_i = \alpha \text{ vs } \data_i = \perp$  & A specific class of distributions representing the adversary's uncertainty about the data. & No        \\ \bottomrule
\end{tabular}
\end{table*}

\section{Pufferfish Mechanisms}

Algorithm ~\ref{alg:mqm_lap} and ~\ref{alg:mqm_exp} are details for Pufferfish Laplacian Mechanism and Pufferfish Exponential Mechanism. 





\begin{algorithm}[]
\caption{Pufferfish Laplace Mechanism}
\label{alg:mqm_lap}
\begin{algorithmic}[1]
\REQUIRE Input dataset $\data \sim \Data$ and indices set $I$, a collection of $(a, b)$-influence $\mathcal{F}_{\mathrm{infl}}(\priorset, \secretset)$,  query $F$ with Lipschitz constant $L$,  Pufferfish privacy parameter $\epsPuffer$

\STATE Initialize $\epsDP \gets \infty$

\FOR{each $(a_\ell, b_\ell)$-influence pair in $\mathcal{F}_{\mathrm{infl}}(\priorset, \secretset)$}
    \IF{$a_\ell \geq \epsPuffer$}
        \STATE $\epsDP_{\ell} = \epsPuffer/|I|$
    \ELSE
        \STATE Compute $\epsDP_{\ell} = \frac{\epsPuffer - a}{b}$
    \ENDIF
    \STATE $\epsDP \gets \min(\epsDP, \epsDP_{\ell})$
\ENDFOR

\STATE Sample $\eta \sim \mathrm{Lap}(L/\epsDP)$
\RETURN $Z := F(\data) + \eta$
\end{algorithmic}
\end{algorithm}



\begin{algorithm}[]
\caption{Pufferfish Exponential Mechanism}
\label{alg:mqm_exp}
\begin{algorithmic}[1]
\REQUIRE Input dataset $\data \sim \Data$ and indices set $I$, a collection of $(a, b)$-influence $\mathcal{F}_{\mathrm{infl}}(\priorset, \secretset)$, utility function $u(\data, r)$ with Lipschitz constant $L$,  Pufferfish privacy parameter $\epsPuffer$

\STATE Initialize $\epsDP \gets \infty$

\FOR{each $(a_\ell, b_\ell)$-influence pair in $\mathcal{F}_{\mathrm{infl}}(\priorset, \secretset)$}
    \IF{$a_\ell \geq \epsPuffer$}
        \STATE $\epsDP_{\ell} = \epsPuffer/|I|$
    \ELSE
        \STATE Compute $\epsDP_{\ell} = \frac{\epsPuffer - a}{b}$
    \ENDIF
    \STATE $\epsDP \gets \min(\epsDP, \epsDP_{\ell})$
\ENDFOR

\STATE Sample $Z \in \mathcal{R}$ with probability proportional to $\exp\left( \frac{u(\data, r)\epsDP}{2 L} \right)$
\RETURN $Z$
\end{algorithmic}
\end{algorithm}

\section{Proofs}\label{sec:proofs}

\printProofs

\section{Experimental Setup}
\label{appendix:exp}

\subsection{Dataset and Preprocessing}
\textbf{User Check-in Experiment: }We conduct experiments on the \textbf{Foursquare check-in dataset}, which records user trajectories over time and space across 77 countries.  Each venue is labeled with a fine-grained category (e.g., \textit{Cocktail Bar}, \textit{Theme Park}, \textit{Hostel}).  Since the raw dataset contains more than 400 distinct location labels, we manually curated a unified taxonomy of 77 common labels shared across all countries.  All the other labels were merged into the nearest semantic commonly shared category (for example, \textit{Cocktail Bar} $\rightarrow$ \textit{Bar}), while venues that could not be confidently mapped were assigned to an additional category \textit{``other''}.  
A detailed mapping table describing the merge rules is provided in Appendix~\ref{appendix:labelmap}.
\textbf{Physical Activity Experiment: } We use the \emph{Capture24} dataset~\cite{chan2021capture,chan2024capture}, which contains wrist-worn accelerometer recordings collected in daily living from 151 participants in the Oxfordshire area (2014–2016). Each participant wore the device for roughly 24 hours. Multiple labeling schemes are available; we adopt \texttt{WillettsMET2018}, which provides 11 activity categories designed around metabolic-equivalent (MET)–related patterns, including the following activities: bicycling, gym, sitstand+activity, sitstand+lowactivity, sitting, sleep, sports, standing, vehicle, walking, walking+activity.

\subsection{Smoothing and Stability}
To ensure numerical stability and handle rare transitions, we apply additive smoothing with tolerance $\tau = 10^{-5}$:
\begin{enumerate}
    \item Add $\tau$ to every zero entry in each row of $P$.
    \item Subtract the total added mass proportionally from non-zero entries so that each row remains normalized ($\sum_j P_{ij}=1$) and all entries stay non-negative.
\end{enumerate}
This preserves stochasticity while avoiding zero-probability states during evaluation.

\paragraph{An example of smoothing transition matrix} Given a transition matrix as the following:

\begin{lstlisting}
[[0.5    0.     0.1666 0.0555 0.2777]
 [0.1428 0.3571 0.1428 0.0714 0.2857]
 [0.0361 0.0120 0.4698 0.1325 0.3493]
 [0.     0.     0.0897 0.4615 0.4487]
 [0.0125 0.0251 0.1006 0.0911 0.7704]]
\end{lstlisting}

We first add one additional row/column to include the state of ``others'':

\begin{lstlisting}
[[0.5    0.     0.1666 0.0555 0.2777 0.    ]
 [0.1428 0.3571 0.1428 0.0714 0.2857 0.    ]
 [0.0361 0.0120 0.4698 0.1325 0.3493 0.    ]
 [0.     0.     0.0897 0.4615 0.4487 0.    ]
 [0.0125 0.0251 0.1006 0.0911 0.7704 0.    ]
 [0.2    0.2    0.2    0.2    0.2    0.    ]]
\end{lstlisting}

And add smoothing to the transition matrix:

\begin{lstlisting}
[[0.49999 1e-05   0.16666 0.05555 0.27777 1e-05]
 [0.14286 0.35714 0.14286 0.07143 0.28571 1e-05]
 [0.03614 0.01205 0.46987 0.13253 0.34939 1e-05]
 [1e-05   1e-05   0.08974 0.46152 0.44870 1e-05]
 [0.01258 0.02516 0.10063 0.09119 0.77043 1e-05]
 [0.19999 0.19999 0.19999 0.19999 0.19999 1e-05]]
\end{lstlisting}

\subsection{Label Merging Rules}
\label{appendix:labelmap}

To unify over 400 raw Foursquare location labels across different countries,
we manually curated a set of 78 \textit{common labels}.
Each common label aggregates multiple semantically similar subcategories
(e.g., ``Cocktail Bar'' $\rightarrow$ ``Bar'', ``Hostel'' $\rightarrow$ ``Hotel'').
The following listing summarizes all mapping groups used in our experiments.

\begin{small}
\begin{description}
    \item[\textbf{Nightclub}] Nightclub, Jazz Club, Nightlife Spot, Other Nightlife, Rock Club, Strip Club
    \item[\textbf{Sushi Restaurant}] Sushi Restaurant
    \item[\textbf{Restaurant}] Restaurant, Afghan Restaurant, African Restaurant, American Restaurant, Arepa Restaurant, Argentinian Restaurant, Australian Restaurant, Brazilian Restaurant, Breakfast Spot, Cajun / Creole Restaurant, Caribbean Restaurant, College Cafeteria, Cuban Restaurant, Dim Sum Restaurant, Dumpling Restaurant, Eastern European Restaurant, Ethiopian Restaurant, Falafel Restaurant, Filipino Restaurant, French Restaurant, German Restaurant, Gluten-free Restaurant, Greek Restaurant, Indian Restaurant, Indonesian Restaurant, Japanese Restaurant, Korean Restaurant, Latin American Restaurant, Malaysian Restaurant, Mediterranean Restaurant, Mexican Restaurant, Middle Eastern Restaurant, Molecular Gastronomy Restaurant, Mongolian Restaurant, Moroccan Restaurant, New American Restaurant, Paella Restaurant, Peruvian Restaurant, Portuguese Restaurant, Ramen / Noodle House, Salad Place, Scandinavian Restaurant, Seafood Restaurant, Soup Place, South American Restaurant, Southern / Soul Food Restaurant, Spanish Restaurant, Steakhouse, Swiss Restaurant, Tapas Restaurant, Thai Restaurant, Turkish Restaurant, Vegetarian / Vegan Restaurant, Vietnamese Restaurant
    \item[\textbf{Coffee Shop}] Coffee Shop, Tea Room
    \item[\textbf{Bus Station}] Bus Station, Airport, Airport Gate, Airport Lounge, Airport Terminal, Airport Tram, Boat or Ferry, Bus Line, Plane, Subway, Train, Train Station
    \item[\textbf{Women's Store}] Women's Store, Men's Store
    \item[\textbf{Mall}] Mall, Farmers Market, Fish Market, Flea Market, Market, Shop \& Service, Thrift / Vintage Store
    \item[\textbf{General Travel}] General Travel, Ferry, Light Rail, Taxi, Travel \& Transport, Travel Agency
    \item[\textbf{Lounge}] Lounge, Comedy Club, Roof Deck, Travel Lounge
    \item[\textbf{Professional \& Other Places}] Professional \& Other Places, Financial or Legal Service
    \item[\textbf{Gym}] Gym, Climbing Gym
    \item[\textbf{Hardware Store}] Hardware Store
    \item[\textbf{Department Store}] Department Store
    \item[\textbf{Sporting Goods Shop}] Sporting Goods Shop
    \item[\textbf{BBQ Joint}] BBQ Joint, Wings Joint
    \item[\textbf{Mobile Phone Shop}] Mobile Phone Shop
    \item[\textbf{Bakery}] Bakery
    \item[\textbf{Plaza}] Plaza
    \item[\textbf{Government Building}] Government Building, Capitol Building, Courthouse, Embassy / Consulate, Fire Station, Military Base, Police Station, Radio Station
    \item[\textbf{Sandwich Place}] Sandwich Place, Burrito Place
    \item[\textbf{Design Studio}] Design Studio, Yoga Studio
    \item[\textbf{Fast Food Restaurant}] Fast Food Restaurant, Bagel Shop, Fish \& Chips Shop, Fried Chicken Joint, Mac \& Cheese Joint
    \item[\textbf{Park}] Park, Beach, Castle, Garden, Garden Center, Lake, Mountain, Rest Area, Sculpture Garden, Skate Park, Theme Park, Theme Park Ride / Attraction, Vineyard, Water Park
    \item[\textbf{Medical Center}] Medical Center
    \item[\textbf{City}] City, Harbor / Marina, Island, Road
    \item[\textbf{Cafe}] Cafe, Deli / Bodega, Gaming Cafe, Internet Cafe
    \item[\textbf{Electronics Store}] Electronics Store, Camera Store, Photography Lab, Record Shop, Video Game Store, Video Store
    \item[\textbf{General Entertainment}] General Entertainment, Aquarium, Arcade, Art Gallery, Art Museum, Arts \& Crafts Store, Arts \& Entertainment, Athletic \& Sport, Athletics \& Sports, Boarding House, Bowling Alley, Casino, Concert Hall, Dance Studio, Farm, Hiking Trail, Historic Site, History Museum, Hot Spring, Hunting Supply, Martial Arts Dojo, Monument / Landmark, Museum, Opera House, Performing Arts Venue, Pier, Planetarium, Public Art, River, Rock Climbing Spot, Scenic Lookout, Science Museum, Ski Lodge, Surf Spot, Tourist Information Center, Zoo
    \item[\textbf{Bookstore}] Bookstore
    \item[\textbf{Home (private)}] Home (private)
    \item[\textbf{Building}] Building, Platform
    \item[\textbf{Chinese Restaurant}] Chinese Restaurant, Asian Restaurant
    \item[\textbf{Event Space}] Event Space, Convention Center, Fair, Music Venue, Tech Startup, Voting Booth
    \item[\textbf{Food Court}] Food Court, Airport Food Court, Food, Food \& Drink Shop, Food Truck
    \item[\textbf{Pizza Place}] Pizza Place, Taco Place
    \item[\textbf{Drugstore / Pharmacy}] Drugstore / Pharmacy
    \item[\textbf{Boutique}] Boutique
    \item[\textbf{University}] University, Community College, Law School, Medical School, Music School, Nursery School
    \item[\textbf{Automotive Shop}] Automotive Shop, Bike Rental / Bike Share, Bike Shop, Car Dealership, Motorcycle Shop, Rental Car Location
    \item[\textbf{Hospital}] Hospital, Emergency Room, Eye Doctor, Veterinarian
    \item[\textbf{Office}] Office, Campaign Office, Conference Room, Lighthouse, Meeting Room, Real Estate Office
    \item[\textbf{Convenience Store}] Convenience Store, Accessories Store, Health Food Store, Kids Store, Toy / Game Store
    \item[\textbf{Spa / Massage}] Spa / Massage
    \item[\textbf{Gas Station / Garage}] Gas Station / Garage, EV Charging Station, Storage Facility
    \item[\textbf{Housing Development}] Housing Development
    \item[\textbf{Coworking Space}] Coworking Space, Non-Profit
    \item[\textbf{General College \& University}] General College \& University, Auditorium, College \& University, College Academic Building, College Administrative Building, College Arts Building, College Auditorium, College Baseball Diamond, College Basketball Court, College Bookstore, College Classroom, College Communications Building, College Cricket Pitch, College Engineering Building, College History Building, College Lab, College Library, College Math Building, College Quad, College Science Building, College Technology Building, Laboratory, Library
    \item[\textbf{Pool}] Pool, Gym Pool, Pool Hall
    \item[\textbf{Bank}] Bank, Credit Union
    \item[\textbf{Gym / Fitness Center}] Gym / Fitness Center, College Gym, Outdoors \& Recreation, Tennis, Tennis Court, Volleyball Court
    \item[\textbf{Clothing Store}] Clothing Store, Bridal Shop, Jewelry Store, Laundry Service, Lingerie Store, Shoe Store, Tailor Shop
    \item[\textbf{Diner}] Diner, Cafeteria
    \item[\textbf{Furniture / Home Store}] Furniture / Home Store, Paper / Office Supplies Store
    \item[\textbf{Dessert Shop}] Dessert Shop, Candy Store, Cheese Shop, Cupcake Shop, Donut Shop, Frozen Yogurt, Ice Cream Shop, Snack Place, Yogurt
    \item[\textbf{Grocery Store}] Grocery Store, Butcher, Gourmet Shop
    \item[\textbf{Other Great Outdoors}] Other Great Outdoors, Trail, Trails, Volcano, Volcanoes
    \item[\textbf{Burger Joint}] Burger Joint, Hot Dog Joint
    \item[\textbf{City Hall}] City Hall
    \item[\textbf{Factory}] Factory, Recycling Facility
    \item[\textbf{Student Center}] Student Center, College Rec Center, Fraternity House, Sorority House
    \item[\textbf{Miscellaneous Shop}] Miscellaneous Shop
    \item[\textbf{Cosmetics Shop}] Cosmetics Shop
    \item[\textbf{Church}] Church, Mosque, Shrine, Spiritual Center, Synagogue, Temple
    \item[\textbf{Bridge}] Bridge
    \item[\textbf{Dentist's Office}] Dentist's Office
    \item[\textbf{Movie Theater}] Movie Theater, College Theater, Indie Movie Theater, Indie Theater, Multiplex, Theater
    \item[\textbf{Neighborhood}] Neighborhood
    \item[\textbf{Doctor's Office}] Doctor's Office
    \item[\textbf{Field}] Field, Baseball Field, Baseball Stadium, Basketball Court, Basketball Stadium, College Football Field, College Hockey Rink, College Soccer Field, College Stadium, College Tennis Court, Football Stadium, Golf Course, Hockey Arena, Hockey Field, Paintball Field, Skating Rink, Ski Area, Ski Chairlift, Ski Chalet, Soccer Field, Soccer Stadium, Stadium, Track Stadium
    \item[\textbf{Post Office}] Post Office
    \item[\textbf{Italian Restaurant}] Italian Restaurant
    \item[\textbf{Hotel}] Hotel, Bed \& Breakfast, Hostel, Hotel Bar, Hotel Pool, Motel, Resort
    \item[\textbf{Playground}] Playground, Campground, College Track, Cricket Ground, Dog Run, Racetrack, Ski Trail, Track
    \item[\textbf{Salon / Barbershop}] Salon / Barbershop, Nail Salon, Tanning Salon
    \item[\textbf{Residential Building (Apartment / Condo)}] Residential Building (Apartment / Condo), Assisted Living, College Residence Hall, Residence
    \item[\textbf{High School}] High School, Elementary School, Middle School, School
    \item[\textbf{Bar}] Bar, Apres Ski Bar, Beer Garden, Brewery, Cocktail Bar, Distillery, Dive Bar, Gastropub, Gay Bar, Hookah Bar, Juice Bar, Karaoke Bar, Liquor Store, Piano Bar, Pub, Sake Bar, Speakeasy, Sports Bar, Whisky Bar, Wine Bar, Wine Shop, Winery
    \item[\textbf{other}] Animal Shelter, Antique Shop, Board Shop, Car Wash, Cemetery, Daycare, Flower Shop, Funeral Home, Gift Shop, Hobby Shop, Moving Target, Music Store, Newsstand, Optical Shop, Parking, Pet Service, Pet Store, Smoke Shop, Stable, Stables, Tattoo Parlor, Trade School, Well
\end{description}
\end{small}

}

\end{document}